\documentclass[letterpaper,11pt]{article}

\usepackage{amsmath}
\usepackage{amsfonts}
\usepackage{amssymb}
\usepackage{amsthm}
\usepackage[usenames]{color}
\usepackage{fullpage}
\usepackage{xspace}
\usepackage{url,ifthen}
\usepackage{srcltx}
\usepackage{multirow}
\usepackage{boxedminipage}
\usepackage[margin=1in]{geometry}
\usepackage{nicefrac}
\usepackage{xspace}
\usepackage{graphicx}
\usepackage[usenames]{color}
\usepackage{srcltx}
\definecolor{DarkGreen}{rgb}{0.1,0.5,0.1}
\definecolor{DarkRed}{rgb}{0.5,0.1,0.1}
\definecolor{DarkBlue}{rgb}{0.1,0.1,0.5}
\usepackage[small]{caption}

\usepackage[pdftex]{hyperref}
\hypersetup{
    unicode=false,          
    pdftoolbar=true,        
    pdfmenubar=true,        
    pdffitwindow=false,      
    pdfnewwindow=true,      
    colorlinks=true,       
    linkcolor=DarkRed,          
    citecolor=DarkGreen,        
    filecolor=DarkRed,      
    urlcolor=DarkBlue,          
    pdftitle={Private Data Release Via Learning Thresholds},
    pdfauthor={Moritz Hardt, Guy Rothblum, Rocco Servedio},
}

\def\draft{1}

\def\submit{0}

\ifnum\draft=1 
    \def\ShowAuthNotes{1}
\else
    \def\ShowAuthNotes{0}
\fi

\ifnum\submit=1
\newcommand{\forsubmit}[1]{#1}
\newcommand{\forreals}[1]{}
\else
\newcommand{\forreals}[1]{#1}
\newcommand{\forsubmit}[1]{}
\fi

\ifnum\ShowAuthNotes=1
\newcommand{\authnote}[2]{{ \footnotesize \bf{{\color{DarkRed}[#1's Note:}
{\color{DarkBlue}#2}]}}}
\else
\newcommand{\authnote}[2]{}
\fi

\newcommand{\myparagraph}{\paragraph}

\newtheorem{theorem}{Theorem}[section]
\newtheorem{remark}[theorem]{Remark}
\newtheorem{lemma}[theorem]{Lemma}
\newtheorem{corollary}[theorem]{Corollary}

\newtheorem{claim}[theorem]{Claim}

\theoremstyle{definition}
\newtheorem{definition}[theorem]{Definition}

\newtheorem{openQ}[theorem]{Open Question}

\usepackage[T1]{fontenc}
\usepackage{times}
\usepackage[varg]{txfonts} 
\renewcommand{\mathbb}{\varmathbb}
\usepackage{microtype}

\newcommand{\Psymb}{\mathbb{P}}

\DeclareMathOperator*{\ProbOp}{\Psymb}
\renewcommand{\Pr}{\ProbOp}

\newcommand{\nfrac}{\nicefrac}

\newcommand{\mper}{\,.}
\newcommand{\mcom}{\,,}

\newcommand{\Le}{{\cal L}}
\newcommand{\Q}{{\cal Q}}
\newcommand{\A}{{\cal A}}
\newcommand{\F}{{\cal F}}
\newcommand{\U}{{\cal U}}

\newcommand{\GQ}{{\cal G \cal Q}}
\newcommand{\DB}{{\mathit{DB}}}
\renewcommand{\DB}{D}
\newcommand{\zo}{\{0,1\}}

\newcommand{\cF}{{\cal F}}

\newcommand{\cQ}{{\cal Q}}

\newcommand{\defeq}{\stackrel{\small \mathrm{def}}{=}}
\renewcommand{\leq}{\leqslant}
\renewcommand{\le}{\leqslant}
\renewcommand{\geq}{\geqslant}
\renewcommand{\ge}{\geqslant}

\newcommand{\Set}[1]{\left\{#1\right\}}

\newcommand{\bits}{\{0,1\}}

\usepackage{bm}

\newcommand{\sign}{\mathit{sign}}
\newcommand{\polylog}{{\rm polylog}}
\newcommand{\poly}{{\rm poly}}

\renewcommand{\epsilon}{\varepsilon}

\newcommand{\eps}{\epsilon}

\newcommand{\Laplace}{\mathrm{Lap}}

\newcommand{\remove}[1]{}

\newcommand{\Lap}{\Laplace}
\newcommand{\ignore}[1]{}
\newcommand{\TO}{\mathcal{TO}}
\newcommand{\Mid}{\ensuremath{~\Big|~}}

\newenvironment{itm}
{\begin{itemize}
}{\end{itemize}}

\newenvironment{enum}
{\begin{enumerate}
}{\end{enumerate}}

\newcommand{\PrivLearn}{\textsc{PrivLearn}\xspace}
\newcommand{\base}{\mathrm{base}}
\newcommand{\iter}{\mathrm{iter}}
\newcommand{\total}{\mathrm{total}}
\newcommand{\AC}{\mathrm{AC}}

\title{\bf Private Data Release via Learning Thresholds}
\author{Moritz Hardt\thanks{Center for Computational Intractability, Department
of Computer Science, Princeton University.  Supported by NSF grants
CCF-0426582 and CCF-0832797. Email: {\tt mhardt@cs.princeton.edu}.}
\and Guy N. Rothblum\thanks{Microsoft Research, Silicon Valley Campus. Most of this work was done while the author was at the Department of Computer Science at Princeton University and Supported by NSF Grant CCF-0832797 and by a Computing
Innovation Fellowship. Email: {\tt rothblum@alum.mit.edu}.} \and Rocco A. Servedio\thanks{Columbia University Department of Computer
Science and Center for Computational Intractability, Department of Computer Science, Princeton
University.  Supported by NSF grants CCF-0832797, CNS-07-16245 and CCF-0915929. Email: {\tt rocco@cs.columbia.edu}}
}
\begin{document}
\maketitle

\begin{abstract}

This work considers {\em computationally efficient} privacy-preserving
data release. We study the task of analyzing a database containing
sensitive information about individual participants. Given a set of
statistical queries on the data, we want to release approximate
answers to the queries while also guaranteeing {\em differential
privacy}---protecting each participant's sensitive data.

Our focus is on computationally efficient data release algorithms; we
seek algorithms whose running time is polynomial, or at least
sub-exponential, in the data dimensionality. Our primary contribution
is a {\em computationally efficient} reduction from differentially
private data release for a class of counting queries, to learning
thresholded sums of predicates from a related class.

We instantiate this general reduction with a variety of algorithms for
learning thresholds. These instantiations yield several new results
for differentially private data release. As two examples, taking
$\zo^d$ to be the data domain (of dimension $d$), we obtain
differentially private algorithms for:

\begin{enumerate}

\item

Releasing all $k$-way conjunction counting queries (or $k$-way
contingency tables). For any given~$k$, the resulting data release
algorithm has bounded error as long as the database is of size at least
$d^{O\big(\sqrt{k\log(k\log d)}\big)}$ (ignoring the dependence on other parameters). The running time is polynomial in the database size.
The best sub-exponential time algorithms known prior to our work required a database of size $\tilde{O}(d^{k/2})$ [Dwork McSherry Nissim and Smith 2006].


\item

Releasing a $(1 - \gamma)$-fraction of all $2^d$ parity counting
queries. For any $\gamma \geq \poly(1/d)$, the algorithm has bounded error as long as the database is of size at least $\poly(d)$ (again ignoring the dependence on other
parameters). The running time is polynomial in the database size.



\end{enumerate}

Several other instantiations yield further results for
privacy-preserving data release. Of the two results highlighted above,
the first learning algorithm uses techniques for representing
thresholded sums of predicates as low-degree polynomial threshold
functions.  The second learning algorithm is based on  Jackson's
Harmonic Sieve algorithm [Jackson 1997]. It utilizes Fourier analysis
of the database viewed as a function mapping queries to answers.

\end{abstract}

\vfill
\thispagestyle{empty}
\setcounter{page}{0}
\pagebreak

%
%

\section{Introduction}
This work considers privacy-preserving statistical analysis of sensitive
data. In this setting, we wish to extract statistics from a database $D$
that contains information about $n$ individual participants. Each
individual's data is a record in the {\em data domain} $\U$. We focus here
on the offline (or non-interactive) setting, in which the information
to be extracted is specified by a set $\Q$ of statistical queries. Each
query $q \in \Q$ is a function mapping the database to a query answer,
where in this work we focus on real-valued queries with range $[0,1]$. Our
goal is {\em data release}: Extracting approximate answers to all the
queries in the query set $\Q$.

An important concern in this setting is protecting the privacy of
individuals whose sensitive data (e.g. medical or financial records) are being
analyzed. Differential privacy \cite{DworkMNS06} provides a rigorous notion of
privacy protection, guaranteeing that each individual only has a small
effect on the data release algorithm's output. A growing body of work
explores the possibility of extracting rich statistics in a differentially
private manner. One line of research~\cite{BlumLR08,DworkNRRV09,DworkRV10,RothR10,HardtR10}
has shown that differential
privacy often permits surprisingly accurate statistics.
These works put forward
general algorithms and techniques for differentially private data
analysis, but the algorithms have running time that is (at least)
exponential in the dimensionality of the data domain.
Thus, a central question
in differentially private data analysis is to develop general techniques
and algorithms that are {\em efficient}, i.e. with running time that is
polynomial (or at least sub-exponential) in the data dimensionality. While
some computational hardness results are known
\cite{DworkNRRV09,UllmanV11,GuptaHRU11}, they apply only to restricted
classes of data release algorithms.

\myparagraph{This Work.} Our primary contribution is a computationally
efficient new tool for privacy-preserving data release: a general
reduction to the task of \emph{learning thresholds of sums of predicates}. The
class of predicates (for learning) in our reduction is derived directly from the class of
queries (for data release).

At a high level, we draw a connection between data release and
learning as follows. In the data release setting, one can view the {\em
database as a function}: it maps queries in $\Q$ to answers in $[0,1]$.
The data release goal is approximating this function on queries/examples
in~$\Q$. The challenge is doing so with only bounded access to the
database/function; in particular, we only allow access that preserves
differential privacy. For example, this often means that we only get a bounded number of oracle queries to the
database function with noisy answers.

At this high level there is a striking similarity to learning theory,
where a standard goal is to efficiently learn/approximate a function given
limited access to it, e.g. a bounded number of labeled examples or oracle
queries. Thus a natural approach to data release is \emph{learning the
database function} using a computational learning algorithm.

While the
approach is intuitively appealing at this high level, it faces immediate
obstacles because of apparent incompatibilities between the requirements of learning algorithms
and the type of ``limited'' access to data that are imposed by private data release.
For example, in the data
release setting a standard technique for ensuring differential privacy is
adding noise, but many efficient learning algorithms fail badly when run on
noisy data.  As another example, for private data release, the number of (noisy) database accesses is often very restricted: e.g sub-linear, or at most quadratic in the database size. In the learning setting, on the other hand, it is almost always
the case that the number of examples or oracle queries required to learn a
function is \emph{lower bounded} by its description length (and is often a large polynomial in the description length).

Our work explores the connection between learning and private data release.
We
\begin{itemize}
\item[(i)] give an efficient reduction that shows that, in
fact, a general class of data release tasks \emph{can} be reduced to related
and natural computational learning tasks; and 
\item[(ii)] instantiate this
general reduction using new and known learning algorithms 
to obtain new computationally
efficient differentially private data release algorithms.
\end{itemize}

Before giving more details on our reduction in Section~\ref{subsec:intro-reduction},
we briefly discuss its context and some of the ways that we apply/instantiate it.
While the search for efficient differentially private data release algorithms is relatively new,
there are decades of work in learning theory aimed at developing techniques
and algorithms for computationally efficient learning, going back to the
early work of Valiant \cite{Valiant84}.  Given the high-level similarity between the two fields,
leveraging the existing body of work and insights from learning theory for data release is a promising direction for future research; we view our reduction as a step in this
direction. We note that our work is by no means the first to draw a
connection between privacy-preserving data release and learning theory;
as discussed in the ``Related Work'' section below, several prior works used
learning techniques in the data release setting.
A novelty in our work is that it gives an explicit and
modular reduction from data release to natural learning problems.
Conceptually, our reduction overcomes two main hurdles:
\begin{itemize}
\item
bridging the gap between the \emph{noisy} oracle access arising in private data release and the \emph{noise-free} oracle access required by many learning algorithms (including the ones we use).
\item
avoiding any dependence on the database size in the complexity of
the learning algorithm being used.
\end{itemize}

We use this reduction to construct new data
release algorithms. In this work we explore two main
applications of our reduction.
The first aims to answer boolean conjunction queries (also
known as contingency tables or marginal queries), one of the most
well-motivated and widely-studied classes of statistical queries in the
differential privacy literature. Taking the data universe $\U$ to be
$\zo^d$, the $k$-way boolean conjunction corresponding to a subset $S$ of
$k$ attributes in $[d]$ counts what fraction of items in the
database have all the attributes in $S$ set to 1. Approximating the
answers for $k$-way conjunctions (or all conjunctions) has been the
focus of several past works (see, e.g. \cite{BarakCDKMT07,KasiRSU10,UllmanV11,GuptaHRU11}).
Applying our reduction with a new learning algorithm tailored for this
class, we obtain a data release algorithm that, for databases of size
$d^{O\big(\sqrt{k\log(k\log d)}\big)}$, releases accurate answers to all $k$-way
conjunctions simultaneously (we ignore for now the dependence of the database
size on other
parameters such as the error).
The running time is $\poly(d^k).$
Previous algorithms either had running time $2^{\Omega(d)}$ (e.g.
\cite{DworkNRRV09})
or required a database of size $d^{k/2}$ (adding independent noise
\cite{DworkMNS06}). We also obtain better bounds for the task of
approximating the answers to a large fraction of {\em all} (i.e. $d$-way)
conjunctions under arbitrary distributions. These results follow from
algorithms for learning thresholds of sums of the relevant predicates; we base
these algorithms on learning theory techniques for representing such functions as
low-degree polynomial threshold functions, following works such
as~\cite{KlivansS04,KlivansOS04}.
We give an overview of these results in Section~\ref{subsec:intro-conjunctions}
below.

Our second application uses Fourier analysis of the database (viewed,
again, as a real-valued function on the queries in $\Q$). We obtain new
polynomial and quasi-polynomial data release algorithms for parity
counting queries and low-depth ($\AC^0$) counting queries respectively. The
learning algorithms we use for this are (respectively) Jackson's Harmonic
Sieve algorithm \cite{Jackson97}, and an algorithm for
learning Majority-of-$\AC^0$ circuits due to Jackson \emph{et al.}
\cite{JacksonKS02}. We elaborate on these results in Section \ref{subsec:intro-Fourier} below.

\subsection{Private Data Release Reduces to Learning Thresholds}
\label{subsec:intro-reduction}

In this section we give more details on the reduction from
privacy-preserving data release to learning thresholds. The full details
are in Sections~\ref{sec:privacy-via-learning} and~\ref{sec:mainproof}. We begin with loose
definitions of the data release and learning tasks we consider, and then
proceed with (a simple case of) our reduction.

\medskip

\noindent {\bf Counting Queries, Data Release and Learning Thresholds.}
We begin with preliminaries and an informal specification of the data
release and learning tasks we consider in our reduction (see Sections
\ref{sec:prelim} and \ref{sec:setup} for full definitions). We refer to an element $u$
in data domain $\U$ as an \emph{item}. A
\emph{database} is a collection of $n$ items from $\U$. A {\em counting query} is
specified by a predicate $p: \U \rightarrow \zo$, and the query $q_p$ on
database $D$ outputs the fraction of items in $D$ that satisfy $p$, i.e.
$\frac{1}{n} \sum_{i=1}^n p(D_i)$. A {\em class} of counting queries is
specified by a set $\Q$ of query descriptions and a predicate $P\colon
\Q\times\U\to\bits$. For a query $q \in \Q$, its corresponding predicate is
$P(q,\cdot): \U \rightarrow \zo$. We will sometimes fix a data item $u \in \U$
and consider the predicate $p_u(\cdot) \triangleq P(\cdot,u) : \Q
\rightarrow \zo$.

Fix a data domain $\U$ and query class $\Q$ (specified by a predicate
$P$). A {\em data release algorithm} $\A$ gets as input a database $D$, and
outputs a {\em synopsis} $S:\Q \rightarrow [0,1]$ that provides
approximate answers to queries in $\Q$.  We say that $\A$ is an
$(\alpha,\beta,\gamma)$ \emph{distribution-free data release algorithm for
$(\U,\Q,P)$} if, for any distribution $G$ over the query set $\Q$, with
probability $1-\beta$ over the algorithm's coins, the synopsis $S$
satisfies that with probability $1-\gamma$ over $q \sim G$, the (additive)
error of $S$ on $q$ is bounded by $\alpha$. Later we will also
consider data release algorithms that only work for a specific
distribution or class of distributions (in this case we will not call the
algorithm distribution-free). Finally, we assume for now that the data
release algorithm only accesses the distribution $G$ by sampling queries
from it, but later we will also consider more general types of access (see
below). A {\em differentially private} data release algorithm is one whose
output distribution (on synopses) is differentially private as per
Definition \ref{def:diffP}. See Definition \ref{def:data-release} for full
and formal details.

Fix a class $\Q$ of {\em examples} and a set $\F$ of predicates on $\Q$.
Let $\F_{n,t}$ be the set of thresholded sums from~$\F$, i.e., the set of
functions of the form $f=\mathbb{I}\left\{\frac1n\sum_{i=1}^nf_i\ge
t\right\}$, where $f_i\in\F$ for all $1\le i\le n.$ We refer to functions
in $\F_{n,t}$ as \emph{$n$-thresholds.} An algorithm for learning
thresholds gets access to a function in $\F_{n,t}$ and outputs a {\em
hypothesis} $h:\Q \rightarrow \zo$ that labels examples in $\Q$. We say
that it is a $(\gamma,\beta)$ distribution-free learning algorithm for
learning thresholds over $(\Q,\F)$ if, for any distribution $G$ over the
set $\Q$, with probability $1-\beta$ over the algorithm's coins the output
hypothesis $h$ satisfies that with probability $1-\gamma$ over $q \sim G$,
$h$ labels $q$ correctly. As above, later we will also consider
learning algorithms that are not distribution free, and only work for a specific distribution or class or
distributions. For
now, we assume that the learning algorithm only accesses the distribution
$G$ by drawing examples from it. These examples are labeled using the
target function that the algorithm is trying to learn. See Definition
\ref{def:learning-thresholds} for full and formal details.

\myparagraph{The Reduction.}
We can now describe (a simple case of) our reduction from
differentially private data release to learning thresholds. For any data
domain $\U$, set $\Q$ of query descriptions, and predicate $P: \Q \times
\U \rightarrow \zo$, the reduction shows how to construct a (distribution free) data
release algorithm given a (distribution free) algorithm for learning
thresholds over $(\Q,\{p_u\colon u\in\U\})$, i.e., any algorithm for
learning thresholds where $\Q$ is the example set and the set ${\cal F}$ of
predicates (over $\Q$) is obtained by the possible ways of fixing the
$u$-input to $P$. The resulting data release algorithm is
$(\alpha,\beta,\gamma)$-accurate as long as the database is not too
small; the size bound depends on the desired accuracy parameters and on
the learning algorithm's sample complexity. The efficiency of the learning algorithm is preserved
(up to mild polynomial factors).

\begin{theorem}[Reduction from Data Release to Learning Thresholds, Simplified] \label{thm:intro} Let $\U$ be a data universe, $\Q$ a set
of query descriptions, and $P\colon\Q \times \U \rightarrow \zo$ a
predicate. There is an $\eps$-differentially private
$(\alpha,\beta,\gamma)$-accurate distribution free data-release algorithm for $(\U,\Q,P)$,
provided that:

\begin{enumerate}

\item there is a distribution-free learning algorithm $\Le$ that
($\gamma$,$\beta$)-learns thresholds over $(\Q,\{p_u\colon u\in\U\})$
using $b(n,\gamma,\beta)$ labeled examples and
running time $t(n,\gamma,\beta)$ for learning $n$-thresholds. \item
$n \ge
\frac{C\cdot b(n',\gamma',\beta') \cdot \log(1/\beta)}{\eps \cdot \alpha \cdot
\gamma}
\mcom$
 where
$n'=\Theta(\log|\cQ|/\alpha^2),$
$\beta'=\Theta(\beta \cdot \alpha),$ $\gamma'=\Theta(\gamma \cdot \alpha),$
$C=\Theta(1).$
\end{enumerate}

Moreover, the data release algorithm only accesses the query distribution
by sampling. The number of samples taken is $O(b(n',\gamma',\beta') \cdot
\log(1/\beta) / \gamma)$ and the running time is
$\poly(t(n',\gamma',\beta'),n,1/\alpha,\log(1/\beta),1/\gamma)$.

\end{theorem}

Section \ref{sec:mainthm} gives a formal (and more general) statement in Theorem \ref{thm:privacy-via-learning}.
Section \ref{sec:proofoverview} gives a proof overview, and Section
\ref{sec:mainproof} gives the full proof. Note that, since the data
release algorithm we obtain from this reduction is distribution free (i.e.
works for any distribution on the query set) and only accesses the query
distribution by sampling, it can be {\em boosted} to yield accurate
answers on {\em all} the queries \cite{DworkRV10}.

\myparagraph{A More General Reduction.}
For clarity of exposition, we gave above a simplified form of the
reduction. This assumed that the learning algorithm is {\em
distribution-free} (i.e. works for any distribution over examples) and
only requires sampling access to labeled examples. These strong assumptions enable us to
get a distribution-free data release algorithm that only accesses the query distribution by sampling.

We also give a reduction that applies even to distribution-specific learning
algorithms that require (a certain kind of) oracle access to the function
being learned.  In addition to sampling labeled examples, the learning
algorithm can: $(i)$ estimate the distribution $G$ on any example $q$ by
querying $q$ and receiving a (multiplicative) approximation to the probability
$G[q]$; and $(ii)$ query an oracle for  the function $f$ being learned on any
$q$ such that $G[q] \neq 0.$ We refer to this as {\em approximate distribution
restricted oracle access}, see Definition
\ref{def:distribution-oracle-access}. Note that several natural learning
algorithms in the literature use oracle queries in this way; in particular, we
show that this is true for Jackson's Harmonic Sieve Algorithm
\cite{Jackson97}, see Section \ref{sec:Fourier}.

Our general reduction gives a data release algorithm for a class $\GQ$ of
distributions on the query set, provided we have a learning algorithm which
can also use approximate distribution restricted oracle access, and which
works for a slightly richer class of distributions $\GQ'$ (a {\em smooth
extension}, see Definition \ref{def:smooth-extension}). Again, several such
algorithms (based on Fourier analysis) are known in the literature; our
general reduction allows us to use them and obtain the new data release
results outlined in Section \ref{subsec:intro-Fourier}.

\myparagraph{Related Work: Privacy and Learning.}
Our new reduction adds to the fruitful and growing interaction between the fields of differentially private data release and learning theory. Prior works also explored this connection. In our work, we ``import'' learning theory techniques by drawing a correspondence between the database (in the data release setting), for which we want to approximate query answers, and the target function (in the learning setting) which labels examples. Several other works have used this correspondence (implicitly or explicitly), e.g. \cite{DworkNRRV09,DworkRV10,GuptaHRU11}.
A different view, in which \emph{queries} in the data release setting correspond to \emph{concepts} in learning theory, was used
in~\cite{BlumLR08} and also in \cite{GuptaHRU11}.

There is also work on \emph{differentially private learning algorithms} in
which the goal is to give differentially private variants of various learning
algorithms~\cite{BlumDMN05,KasiLNRS08}.

\subsection{Applications (Part I): Releasing Conjunctions}
\label{subsec:intro-conjunctions}

We use the reduction of Theorem \ref{thm:intro} to obtain new data release
algorithms ``automatically'' from learning algorithms that satisfy the
theorem's requirements. Here we describe the distribution-free
data release algorithms we obtain for approximating conjunction counting
queries. These use learning algorithms (which are themselves distribution-free
and require only random examples) based on polynomial threshold functions.

Throughout this section we fix the query class under consideration to be
conjunctions. We take $\U=\{0,1\}^d$, and a (monotone) conjunction $q \in \Q =
\zo^d$ is satisfied by $u$ iff $\forall i \mbox{ s.t. } q_i=1$ it is also the
case that $u_i =1$. (Our monotone conjunction results extend easily to general
non-monotone conjunctions with parameters unchanged.\footnote{To see this,
extend the data domain to be $\zo^{2d}$, and for each item in the original
domain include also its negation. General conjunctions in the original data
domain can now be treated as monotone conjunctions in the new data domain.
Note that the locality of a conjunction is unchanged. Our results in this
section are for arbitrary distributions over the set of monotone conjunctions
(over the new domain), and so they will continue to apply to arbitrary
distributions on general conjunctions over the original data domain.})
Our first result is an algorithm for releasing $k$-way conjunctions:

\begin{theorem}[Distribution-Free Data Release for $k$-way conjunctions]
\label{thm:intro-kway-conjunctions}
There is an $\eps$-differentially private $(\alpha,\beta,\gamma)$-accurate
distribution-free data release algorithm, which accesses the query distribution only by sampling,
for the class of $k$-way monotone Boolean conjunction queries.  The algorithm has runtime $\poly(n)$ on databases of size $n$
provided that \[
n \geq
d^{O\left(\sqrt{k\log\left(\frac {k\log d}{\alpha}\right)}\right)}\cdot
\tilde O\left(
\frac{\log\left(\nfrac1\beta\right)^3}
{\epsilon\alpha\gamma^2}\right)\mper
\]
\end{theorem}

Since this is a distribution-free data release algorithm that only accesses
the query distribution by sampling, we can use the boosting results of
\cite{DworkRV10} and obtain a data release algorithm that generates (w.h.p.) a
synopsis that is accurate for {\em all} queries. This increases the running
time to $d^k \cdot \poly(n)$ (because the boosting algorithm needs to
enumerate over all the $k$-way conjunctions). The required bound
on the database size increases slightly but our big-Oh notation hides this small increase.
The corollary is stated formally below:

\begin{corollary}[Boosted Data Release for $k$-way Conjunctions]
\label{cor:boosted-conjunctions}
There is an $\eps$-differentially private $(\alpha,\beta,\gamma = 0)$-accurate
distribution-free data release algorithm for the class of $k$-way monotone
Boolean conjunction queries with runtime $d^k \cdot \poly(n)$
on databases of size $n$, provided that
\[
n \geq
d^{O\left( \sqrt{k\log\left(\frac {k\log d}{\alpha}\right)}\right)}\cdot
\tilde O\left(
\frac{\log\left(\nfrac1\beta\right)^3}
{\epsilon\alpha}\right)\mper
\]
\end{corollary}

We also obtain a new data release algorithm for releasing the answers to {\em all} conjunctions:

\begin{theorem}[Distribution-Free Data Release for {\em All} Conjunctions]
\label{thm:intro-conjunctions}
There is an $\eps$-differentially private $(\alpha,\beta,\gamma)$-accurate
distribution-free data release algorithm, which accesses the query distribution only by sampling, for the class of all monotone Boolean conjunction queries.  The algorithm has runtime $\poly(n)$ on databases of size $n$, provided that
\[
n \geq
d^{O\left(d^{\nfrac13} \cdot \log^{\nfrac23} \left(\frac d{\alpha}\right)\right)}\cdot
\tilde O\left(
\frac{\log\left(\nfrac1\beta\right)^3}
{\epsilon\alpha\gamma^2} \right) \mper
\]
\end{theorem}

Again, we can apply boosting to this result;
this gives improvements over previous work for a certain range of
parameters (roughly $k \in [d^{1/3},d^{2/3}]$).  We omit the details.

\myparagraph{Related Work on Releasing Conjunctions.} Several past works have
considered differentially private data release for conjunctions and $k$-way
conjunctions (also known as marginals and contingency tables). As a corollary of their more general
Laplace and Gaussian mechanisms, the work of Dwork {\em et al.} \cite{DworkMNS06}
showed how to release all $k$-way conjunctions in running time $d^{O(k)}$
provided that the database size is at least $d^{O(k)}$.  Barak {\em et al.} \cite{BarakCDKMT07} showed how to
release {\em consistent} contingency tables with similar database size bounds.
The running time, however, was increased to~$\exp(d).$
We note
that our data-release algorithms do not guarantee consistency. Gupta {\em et
al.}~gave {\em distribution-specific} data release algorithm for $k$-way and for all
conjunctions. These algorithms work for the uniform distribution over ($k$-way
or general) conjunctions. The database size bound  and running time were
(roughly) $d^{\tilde{O}\left(1/{\alpha^2}\right)}$.
For distribution-specific data release on
the uniform distribution, the dependence on $d$ in their work is better than
our algorithms but the dependence on $\alpha$ is worse. Finally, we note that
the general information-theoretic algorithms for differentially private data
release also yield algorithms for the specific case of conjunctions. These
algorithms are (significantly) more computationally expensive, but they have
better database size bounds. For example, the algorithm of~\cite{HardtR10} has
running time $\exp(d)$ but database size bound is (roughly)
$\tilde{O}(d/\alpha^2)$ (for the relaxed notion of
$(\eps,\delta)$-differential privacy).

In terms of negative results, Ullman and Vadhan \cite{UllmanV11} showed that,
under mild cryptographic assumptions, no data release algorithm for
conjunctions (even $2$-way) can output a {\em synthetic database} in running
time less than $\exp(d)$ (this holds even for {\em distribution-specific} data
release on the uniform distribution). Our results side-step this negative
result because the algorithms do not release a synthetic database.

Kasiviswanathan {\em et al.}~\cite{KasiRSU10} showed a lower bound
of $\tilde\Omega\left(\min\left\{d^{
k/2}/\alpha,1/{\alpha^2}\right\}\right)$
on the database size needed for releasing $k$-way conjunctions.
To see that this is consistent with our bounds,
note that our bound on $n$ is always larger than $f(\alpha)=2^{\sqrt{k\log(\nfrac1\alpha)}}/\alpha.$ We have
$f(\alpha)<1/\alpha^2$ only if $k<\log(1/\alpha).$ But in the range where
$k<\log(1/\alpha)$ our theorem needs $n$ to be larger than $d^{k}/\alpha$
which is consistent with the lower bound.

\subsection{Applications (Part II): Fourier-Based Approach}
\label{subsec:intro-Fourier}

We also use Theorem \ref{thm:intro} (in its more general formulation given in
Section \ref{sec:mainthm}) to obtain new data release algorithms for answering
parity counting queries (in polynomial time) and general $\AC^0$ counting queries (in quasi-polynomial time). For both of these we fix the data universe to be $\U=\{0,1\}^d$, and take the set of query descriptions to also be $\Q=\{0,1\}^d$ (with different semantics for queries in the two cases). Both algorithms are distribution-specific, working for the uniform distribution over query descriptions,\footnote{More generally, we can get results for {\em smooth} distributions, we defer these to the full version.} and both instantiate the reduction with learning algorithms that use Fourier analysis of the target function.  Thus the full data release algorithms use Fourier analysis of the database (viewed as a function on queries).

\myparagraph{Parity Counting Queries.} Here we consider counting queries that, for a fixed $q \in \zo^d$, output how many items in the database have inner product 1 with $q$ (inner products are taken over $GF[2]$). I.e., we use the parity predicate $P(q,u)= \sum_i q_i \cdot u_i \text{~(mod 2)}$. We obtain a polynomial-time data release algorithm for this class (w.r.t. the uniform distribution over queries). This uses our reduction, instantiated with Jackson's Harmonic Sieve learning algorithm \cite{Jackson97}.
In Section~\ref{sec:Fourier} we prove:

\begin{theorem}[Uniform Distribution Data Release for Parity Counting Queries.]
\label{thm:intro-parity}
There is an $\eps$-differentially private algorithm for releasing the class of
parity queries over the uniform distribution on $\Q$. For databases of size $n$, the algorithm has runtime $\poly(n)$ and is $(\alpha,\beta,\gamma)$-accurate, provided that
\[
n \geq \frac{\poly(d,\nfrac1\alpha,\nfrac1\gamma,\log(\nfrac1\beta))}{\epsilon}.
\]
\end{theorem}

\myparagraph{$\AC^0$ Counting Queries.} We also consider a quite general class
of counting queries, namely, any query family whose predicate is computed by a
constant depth ($\AC^0$) circuit. For any family of this type, in
Section~\ref{sec:Fourier} we obtain a data release algorithm over the uniform
distribution that requires a database of quasi-polynomial (in
$d$) size (and has running time polynomial in the database size, or
quasi-polynomial in $d$).

\begin{theorem}[Uniform Distribution Data Release for $\AC^0$ Counting Queries]
\label{thm:intro-AC0}

Take $\U=\Q=\zo^d$, and $P(q,u): \Q \times \U \to \{0,1\}$ a predicate computed
by a Boolean circuit of depth~$\ell=O(1)$ and size $\poly(d)$.
There is an $\eps$-differentially private data release algorithm for this query class over the uniform distribution on $\Q$. For databases of size $n$, the algorithm has runtime $\poly(n)$ and is
    $(\alpha,\beta,\gamma)$-accurate, provided that:
\[
n \geq d^{O\left(\log^{\ell}\left(\frac d{\alpha \gamma}\right)\right)} \cdot
\tilde{O}\left( {\frac{\log^3\left(1/\beta\right)}
{\epsilon\alpha^2 \gamma}}
\right)\mper
\]
\end{theorem}

This result uses our reduction instantiated with an algorithm of 
Jackson~\emph{et al.}
\cite{JacksonKS02} for learning Majority-of-$\AC^0$
circuits. To the best of our knowledge, this is the first positive result for
private data release that uses the (circuit) structure of the query class in a
``non black-box'' way to approximate the query answer. We note that the class
of $\AC^0$ predicates is quite rich. For example, it includes conjunctions, approximate
counting \cite{Ajtai83}, and $GF[2]$ polynomials with
$\polylog(d)$ many terms. While our result is
specific to the uniform distribution over $\Q$, we note that
some query sets (and query descriptions) may be amenable to {\em random
self-reducibility}, where an algorithm providing accurate answers to uniformly
random $q \in Q$ can be used to get (w.h.p.) accurate answers to {\em any} $q
\in Q$. We also note that Theorem \ref{thm:intro-AC0} leaves a large degree of freedom in
how a class of counting queries is to be represented. Many different sets of
query descriptions $\Q$ and predicates $P(q,u)$ can correspond to the same set
of counting queries over the same $\U$, and it may well be the case that some
representations are more amenable to computations in $\AC^0$ and/or random
self-reducibility. Finally, we note that the hardness results of Dwork {\em et
al.} \cite{DworkNRRV09} actually considered (and ruled out) efficient
data-release algorithms for $\AC^0$ counting queries (even for the uniform
distribution case), but only when the algorithm's output is a synthetic
database. Theorem \ref{thm:intro-AC0} side-steps these negative results
because the output is not a synthetic database.

\section{Preliminaries} \label{sec:prelim}

\myparagraph{Data sets and differential privacy.}
We consider a data universe $\U$, where throughout this work we take $\U =
\bits^d$. We typically refer to an element $u \in \U$ as an \emph{item}.
A data set (or database) $D$ of size $n$ over the universe $\U$
is an ordered multiset consisting of $n$ items from $\U.$ We will sometimes
think of $D$ as a tuple in $\U^n.$ We use the notation $|D|$ to denote the
size of $D$ (here, $n$).
Two data sets $D,D'$ are called \emph{adjacent} if they are both of size $n$
and they agree in at least $n-1$ items (i.e., their edit distance is at most $1$).

We will be interested in randomized algorithms that map data sets into some abstract
range $\mathcal{R}$ and satisfy the notion of differential privacy.
\begin{definition}[Differential Privacy \cite{DworkMNS06}]
\label{def:diffP}
A randomized algorithm $\mathcal{M}$ mapping data sets over $\U$ to outcomes
in $\mathcal{R}$ satisfies \emph{$(\epsilon,\delta)$-differential privacy} if
for all $S \subset \mathcal{R}$ and every pair of two adjacent databases
$D,D',$ we have $\Pr(\mathcal{M}(D) \in S)\leq e^\epsilon \Pr(\mathcal{M}(D')
\in S) +\delta \mper$ If $\delta=0,$ we say the algorithm satisfies
\emph{$\epsilon$-differential privacy}.
\end{definition}
\myparagraph{Counting queries.}
A class of \emph{counting queries} is specified by a predicate $P\colon
\Q\times\U\to\bits$ where $\Q$ is a set of query descriptions. Each $q \in \Q$ specifies a query  and the answer for a query $q\in\Q$ on a single data
item $u\in\U$ is given by $P(q,u).$
The answer of a counting query $q\in\Q$ on a data set $D$ is defined as
$\frac1n\sum_{u \in D}P(q,u)\mper$

We will often fix a data item $u$ and database $D\in \U^n$ of $n$ data items,
and use the following notation:
\begin{itemize}

 \item $p_u\colon\Q\to\bits,$ $p_u(q) \defeq P(q,u).$ The predicate
on a fixed data item $u$.

 \item $f^D\colon\Q \rightarrow [0,1],$  $f^D(q)\defeq\frac{1}{n}
\sum_{u\in D} P(q,u).$ For an input query description and fixed
database, counts the fraction of database items that satisfy that query.

 \item $f^D_t\colon\Q \rightarrow \bits,$ $f^D_t(q) \defeq \mathbb{I}\left\{
f^D(q) \geq t \right\}$.
For an input query description and fixed database and threshold $t \in
[0,1]$, indicates whether the fraction of database items that satisfy that query
is at least $t$. Here and in the following $\mathbb{I}$ denotes the
$0/1$-indicator function.
\end{itemize}
%


We close this section with some concrete examples of query classes that we will consider.  Fix $\U=\bits^d$ and $\Q=\bits^d.$ The query class of \emph{monotone boolean conjunctions}
is defined by the predicate
%
$P(q,u)=\bigwedge_{i\colon q_i=1} u_i \mper$
Note that we may equivalently write
$P(q,u)= 1-\bigvee_{i\colon u_i=0} q_i\mper$
The query class of \emph{parities over $\{0,1\}^d$} is defined by the predicate
$P(q,u) = \sum_{i : u_i=1} q_i \text{~(mod 2)}\mper$


\section{Private Data Release via Learning Thresholds}
\label{sec:privacy-via-learning}

In this section we describe our reduction from private data release to a
related computational learning task of learning thresholded sums. Section~\ref{sec:setup}
sets the stage, first introducing definitions for handling distributions and access to
an oracle, and then proceeds with notation and formal definitions of
(non-interactive) data release and of learning threshold functions.  Section~\ref{sec:mainthm}
formally states our main theorem giving the reduction, and Section~\ref{sec:proofoverview} gives
an intuitive overview of the proof. The formal proof is then given in
Section~\ref{sec:mainproof}.

\subsection{Distribution access, data release, learning thresholds}
\label{sec:setup}

\begin{definition}[Sampling or Evaluation Access to a Distribution]
\label{def:access-distribution}

Let $G$ be a distribution over a set $\Q$. When we give an algorithm $\A$ {\em
sampling access} to $G$, we mean that $\A$ is allowed to sample items
distributed by $G$. When we give an algorithm $\A$ {\em evaluation access} to
$G$, we mean that $\A$ is both allowed to sample items distributed by $G$ and
also to make oracle queries: in such a query $\A$ specifies any $q \in Q$ and receives back the
probability $G[q] \in [0,1]$ of $q$ under $G$. For both types of access we will
often measure $\A$'s {\em sample complexity} or {\em number of queries} (for the
case of evaluation access).\footnote{Note that, generally speaking, sampling
and evaluation access are incomparably powerful (see
\cite{KearnsMRRSS94,Naor96}). In this work, however, whenever we give an
algorithm evaluation access we will also give it sampling access.}

\end{definition}

\begin{definition}[Sampling Access to Labeled Examples]
\label{def:sampling-labeled-examples} Let $G$ be a distribution over a set
$\Q$ of potential examples, and let $f$ be a function whose domain is $\Q$.
When we give an algorithm $\A$ {\em sampling access to labeled examples by
$(G,f)$}, we mean that $\A$ has sampling access to the distribution
$(q,f(q))_{q \sim G}$.
\end{definition}

\begin{definition}[Data Release Algorithm]
\label{def:data-release}
Fix $\U$ to be a data universe, $\Q$ to be a set of query descriptions, $\GQ$
to be a set of distributions on $\Q$, and $P(q,u): \Q \times \U \rightarrow
\zo$ to be a predicate.
A \emph{$(\U,\Q,\GQ,P)$ data release algorithm} $\A$ is a (probabilistic)
algorithm that gets sampling access to a distribution $G \in \GQ$ and takes as
input accuracy parameters $\alpha,\beta,\gamma > 0$, a database size $n$, and
a database $\DB \in \U^n$.
$\A$ outputs a {\em synopsis} $S:\Q \rightarrow [0,1]$.

We say that $\A$ is \emph{$(\alpha,\beta,\gamma)$-accurate} for databases of
size $n$,
if
for every database $\DB \in \U^n$ and query distribution $G \in \GQ$:
\begin{equation}\label{eq:accuracy1}
\Pr_{S \leftarrow \A(n,\DB,\alpha,\beta,\gamma)} \left[ \Pr_{q \sim G} \left[
|S(q)
- f^D(q) | > \alpha \right] > \gamma \right] < \beta
\end{equation}

We also consider data release algorithms that get evaluation access to $G$. In
this case, we say that $\A$ is a {\em data release algorithm using evaluation
access}. The definition is unchanged, except that $\A$ gets this additional
form of access to $G$.
\end{definition}

When $P$ and $\U$ are understood from the context, we sometimes refer to a $(\U,\Q,\GQ,P)$ data
release algorithm as an \emph{algorithm for releasing the class of queries $\Q$ over $\GQ$}.

This work focuses on {\em differentially private} data release algorithms, i.e. data release algorithms which are $\eps$-differentially private as per Definition \ref{def:diffP} (note that such algorithms must be randomized). In such data release algorithms, the probability of any output synopsis $S$ differs by at most an $e^{\eps}$ multiplicative factor between any two adjacent databases.

We note two cases of particular interest.  The first is when $\GQ$ is the set
of {\em all distributions} over $\Q$. In this case, we say that $\A$ is a {\em distribution-free} data release algorithm. For such algorithms it is possible to apply the ``boosting for queries'' results of
\cite{DworkRV10} and obtain a data release algorithm whose synopsis is
(w.h.p.) accurate on {\em all} queries (i.e. with $\gamma=0$). We note that those
boosting results apply only to data release algorithms that access their
distribution by sampling (i.e. they need not hold for data release algorithms
that use evaluation access).

A second case of interest is when $\GQ$ contains only a single distribution, the
uniform distribution over all queries $\Q.$ In this case both sampling and
evaluation access are easy to simulate.

\begin{remark}
\label{remark:DB-size-vs-error}

Throughout this work, we fix the accuracy parameter $\alpha$, and lower-bound the required database size $n$ needed to ensure the (additive) approximation error is at most $\alpha$. An alternative approach taken in some of the differential privacy literature, is fixing the database size $n$ and upper bounding the approximation error $\alpha$ as a function of $n$ (and of the other parameters). Our database size bounds can be converted to error bounds in the natural way.

\end{remark}

\begin{definition}[Learning Thresholds]
\label{def:learning-thresholds}

Let $\Q$ be a set (which we now view as a domain of potential unlabeled examples) and let $\GQ$ be a set of distributions on
$\Q$. Let $\F$ be a set of
predicates on $\Q$, i.e. functions $Q \to \zo.$ Given $t \in [0,1],$ let $\F_{n,t}$ be the set of all threshold
functions of the form $f=\mathbb{I}\left\{\frac1n\sum_{i=1}^nf_i\ge t\right\}$
where $f_i\in\F$ for all $1\le i\le n.$ We refer to functions in $\F_{n,t}$ as \emph{$n$-thresholds over $\F$}.
Let $\Le$ be a (probabilistic) algorithm that gets sampling access to labeled
examples by a distribution $G \in \GQ$ and a target function $f \in \F_{n,t}$. $\Le$
takes as input accuracy parameters $\gamma,\beta > 0$, an integer $n > 0$, and
a threshold $t \in [0,1]$.  $\Le$ outputs a boolean {\em hypothesis} $h: \Q \rightarrow \zo$.

We say that $\Le$ is an \emph{$(\gamma,\beta)$-learning algorithm for thresholds over
$(\Q,\GQ,\F)$} if for every $\gamma,\beta>0$, every $n$, every $t \in [0,1],$ every $f \in \F_{n,t}$ and every $G \in \GQ$,
we have
\begin{equation}\label{eq:accurate2}
\Pr_{h
\leftarrow \Le(n,t,\gamma,\beta)} \left[ \Pr_{q \sim G} \left[ h(q) \neq f(q)
\right] > \gamma \right] < \beta\mper
\end{equation}
The definition is analogous for all other notions of oracle access (see e.g. Definition \ref{def:distribution-oracle-access} below).

%

\end{definition}

\subsection{Statement of the main theorem} \label{sec:mainthm}

In this section we formally state our main theorem, which establishes a
general reduction from private data release to learning certain threshold
functions. The next definition captures a notion of oracle access
for learning algorithms which arises in the reduction. The definition combines
sampling access to labeled examples with a limited kind of evaluation access
to the underlying distribution and black-box oracle access to the target
function~$f.$
\begin{definition}[approximate distribution-restricted oracle access]
\label{def:distribution-oracle-access}

Let $G$ be a distribution over a domain $\Q$, and let $f$ be a function whose
domain is $\Q$. When we say that an algorithm $\A$ has {\em approximate
$G$-restricted evaluation access to $f$},
we mean that
\begin{enumerate}

 \item $\A$ has sampling access to labeled examples by $(G,f)$; and

 \item $\A$ can make oracle queries on any $q \in Q$, which are answered as follows:  there is
 a fixed constant $c \in [1/3,3]$ such that $(i)$ if $G[q]=0$ the answer is $(0,\perp)$; and $(ii)$ if $G[q] > 0$ the answer is a pair $(c \cdot G[q],f(q))$.

\end{enumerate}
\end{definition}
\begin{remark}
We remark that this is the type of of oracle access provided to the learning
algorithm in our reduction. This is different from the oracle access that the
data release algorithm has.
We could extend Definition \ref{def:data-release} to refer to approximate
evaluation access to~$G$;
all our
results on data release using evaluation access would extend to this weaker
access (under appropriate approximation guarantees). For simplicity, we focus
on the case where the \emph{data release algorithm} has perfectly accurate
evaluation access, since this is sufficient throughout for our purpose.
\end{remark}

One might initially hope that privately releasing a class of queries $\cQ$
over some set of distributions~$\GQ$ reduces to learning corresponding
threshold functions over the \emph{same} set of distributions. However, our
reduction will need a learning algorithm that works for a potentially larger
set of distributions $\GQ'\supseteq\GQ.$ (We will see in
Theorem~\ref{thm:privacy-via-learning} that this poses a stronger requirement
on the learning algorithm.) Specifically, $\GQ'$ will be a \emph{smooth
extension} of $\GQ$ as defined next.

\begin{definition}[smooth extensions]
\label{def:smooth-extension}
Given a distribution $G$ over a set $\Q$ and a value $\mu \geq 1$, the \emph{$\mu$-smooth
extension of $G$} is the set of all distributions $G'$ which are such that $G'[q]
\leq \mu \cdot G[q]$ for all $q \in \Q.$  Given a set of distributions $\GQ$ and $\mu \geq 1$, the
\emph{$\mu$-smooth extension of $\GQ$}, denoted $\GQ'$, is defined as the set of all distributions
that are a $\mu$-smooth extension of some $G\in\GQ.$
\end{definition}

With these two definitions at hand, we can state our reduction in its most
general form. We will combine this general reduction with specific learning
results to obtain concrete new data release algorithms in
Sections~\ref{sec:conj} and~\ref{sec:Fourier}.

\begin{theorem}[Main Result:  Private Data Release via Learning Thresholds]
\label{thm:privacy-via-learning}
\label{thm:main}
Let $\U$ be a data universe, $\Q$ a set of query descriptions, $\GQ$ a set of
distributions over $\Q$, and $P\colon\Q \times \U \rightarrow \zo$ a
predicate.

Then, there is an $\eps$-differentially private
$(\alpha,\beta,\gamma)$-accurate data-release algorithm for databases of size $n$ provided that
\begin{itemize}
\item there is an algorithm $\Le$ that ($\gamma$,$\beta$)-learns thresholds
over $(\Q,\GQ',\{p_u\colon u\in\U\})$, running in time
$t(n,\gamma,\beta)$ and using $b(n,\gamma,\beta)$ queries to an approximate
distribution-restricted evaluation oracle for the target $n$-threshold function,
where $\GQ'$ is the $(2/\gamma)$-smooth
extension of~$\GQ$; and
\item we have
\begin{equation}
\label{eq:req}
n\ge
\frac{C\cdot
b(n',\gamma',\beta')\cdot\log\left(\frac{b(n',\gamma',\beta')}{\alpha\gamma\beta}
\right)
\cdot \log(1/\beta')}{\epsilon\alpha^2\gamma}\mcom
\end{equation}
where $n'=\Theta(\log|\cQ|/\alpha^2),$
$\beta'=\Theta(\beta\alpha),$
$\gamma'=\Theta(\gamma\alpha)$ and $C>0$ is a sufficiently large constant.
\end{itemize}
The running time of the data release algorithm is $\poly(t(n',\gamma',\beta'),n,1/\alpha,\log(1/\beta),1/\gamma)$.
\end{theorem}

The next remark points out two simple modifications of this theorem.
\begin{remark}
\begin{enumerate}
\item We can improve the dependence on~$n$ in~(\ref{eq:req}) by a
factor of~$\Theta(\nfrac1\alpha)$
in the case where the learning algorithm~$\Le$ only uses
sampling access to labeled examples. In this case the data release algorithm also uses only sampling access to the query distribution $G$.
The precise statement is given in
Theorem~\ref{thm:main-random} which we present after the proof of
Theorem~\ref{thm:main}.
\item A similar theorem holds for $(\epsilon,\delta)$-differential
privacy, where the requirement on $n$ in~(\ref{eq:req}) is improved to a
requirement on $\sqrt{n}$ up to a $\log({1}/{\delta})$ factor. The
proof is the same, except for a different (but standard) privacy argument,
e.g., using the Composition Theorem in \cite{DworkRV10}.
\end{enumerate}
\end{remark}

\subsection{Informal proof overview} \label{sec:proofoverview}

Our goal in the data release setting is approximating the query answers $\{
f^D(q) \}_{q \in Q}$. This is exactly the task of approximating or {\em
learning} a sum of $n$ predicates from the set $\F = \{p_u \colon u\in\U\}.$
Indeed, each item $u$ in the database specifies a predicate $p_u$,
and for a fixed query $q \in Q$ we are trying to approximate the sum of the predicates $f^D(q) =
\frac1{|D|}\cdot \sum_{u\in D}p_{u}(q)$. We want to approximate
such a sum in a privacy-preserving manner, and so we will only permit limited
access to the function $f^D$ that we try to approximate. In
particular, we will only allow a bounded number of noisy oracle queries to
this function. Using standard techniques (i.e. adding appropriately scaled
Laplace noise~\cite{DworkMNS06}), an approximation obtained from a bounded number of
noisy oracle queries will be differentially private.
It remains, then, to tackle the task of (i) learning a sum of
$n$ predicates from $\F$ using an oracle to the sum, and (ii) doing so using
only a {\em bounded (smaller than $n$) number} of oracle queries when we are provided {\em
noisy answers}.

\myparagraph{From Sums to Thresholds.}
Ignoring privacy concerns, it is straightforward to reduce the task of learning
a sum $f^D$ of predicates (given an oracle for $f^D$) to the task
of learning thresholded sums of predicates (again given an oracle for $f^D$).
Indeed, set $k=\lceil 3/\alpha\rceil$ and consider the thresholds
$t_1,\dots,t_k$ given by $t_i=i/(k+1).$ Now, given an oracle for $f^D$, it is easy
to simulate an oracle for $f^D_{t_i}$ for any $t_i$. Thus, we can
learn each of the threshold functions $f^D_{t_i}$ to accuracy $1-\gamma/k$ with respect to $G$.
Call the resulting
hypotheses $h_1,\dots,h_k$. Each $h_i$ labels a $(1-\gamma/k)$-fraction of the queries/examples in $Q$ correctly w.r.t the threshold function $f^D_{t_i}$.
We can produce an aggregated hypothesis $h$ for approximating $f^D$ as
follows: given a query/example $q,$ let $h(q)$ equal $t_i$ where $t_i$ is the smallest $i$
such that $h_{i}(q)=0$ and $h_{i+1}(q)=1.$ For random $q~\sim G,$ we will
then have $|h(q)-f^D(q)|\le\alpha/3$ with probability $1-\gamma$ (over the choice of $q$).
%
%

Thus, we have reduced learning a sum to learning thresholded sums (where in both
cases the learning is done with an oracle for the sum).  But because of privacy considerations, we must
address the challenges mentioned above: $(i)$ learning a {\em thresholded}
sum of $n$ predicates using few (less than $n$) oracle queries to the sum, and $(ii)$
learning when the oracle for the sum can return noisy answers. In particular, the noisy sum answers can induce errors on threshold oracle queries (when the sum is close to the threshold).

\myparagraph{Restricting to Large Margins.}
Let us say that a query/example $q \in \Q$ has \emph{low margin with respect to $f^D$ and $t_i$}
if $|f^D(q)-t_i|\le\alpha/7.$
A useful observation is that in the argument sketched
above, we do \emph{not} need to approximate each threshold function
$f_{t_i}^D$ well on low margin elements~$q.$  Indeed, suppose that each
hypothesis $h_{i}$ errs arbitrarily on a set $E_i \subseteq \Q$ that contains
only inputs that have low margin
w.r.t. $f^D$ and $t_i$, but achieves high accuracy $1-\gamma/k$ with respect
to~$G$ conditioned on the event~$\Q\setminus E_i.$
Then the above aggregated hypothesis $h$ would still have high accuracy with high probability over
$q\sim G$; more precisely, $h$ would satisfy $|h(q)-f^D(q)| \leq 2 \alpha/3$
 with probability $1-\gamma$ for $q \sim G.$

The reason is that for every $q\in Q,$ there can only be one
threshold ${i^*}\in\{1,\dots,k\}$ such that $|f^D(q)-t_{i^*}|\le\alpha/7$ (since
any two thresholds are $\alpha/3$- apart from each other). While the threshold hypothesis $h_{i^*}$ might err on $q$ (because $q$ has low margin w.r.t. $t_{i^*}$), the hypotheses $h_{i^*-1}$ and $h_{i^*+1}$ should still be accurate (w.h.p. over $q \sim G$), and thus the aggregated hypothesis $h$ will still output a value between $t_{i^*-1}$ and $t_{i^* +1}$.

\myparagraph{Threshold Access to The Data Set.}
We will use the above observation to our advantage. Specifically, we restrict
all access to the function $f^D$ to what we call a \emph{threshold oracle}.
Roughly speaking, the threshold oracle (which we denote~$\TO$ and define
formally in Section~\ref{sec:oracle}) works as follows:  when given a query
$q$ and a threshold $t$, it draws a suitably scaled Laplacian variable $N$
(used to ensure differential privacy) and returns $1$ if $f^{D}(q)+N\ge
t+\alpha/20$; returns 0 if $f^{D}(q)+N\le t-\alpha/20$; and returns ``$\bot$''
if $t - \alpha/20 < f^D(q)+N < t + \alpha/20.$ If $D$ is large enough then we
can ensure that $|N|\le\alpha/40$ with high probability, and thus whenever the
oracle outputs $\bot$ on a query $q$ we know that $q$ has low margin with
respect to $f^D$ and $t$ (since $\alpha/20 + |N| < \alpha/7$).

We will run the learning algorithm $\Le$ on examples generated using the
oracle $\TO$ after removing all examples for which the oracle
returned $\bot.$ Since we are conditioning on the~$\TO$ oracle not returning $\bot$,
this transforms the distribution $G$ into a conditional distribution
which we denote $G'.$
Since we have only conditioned on removing
low-margin $q$'s, the argument sketched above applies. That is, the
hypothesis that has high accuracy with respect to this
conditional distribution $G'$ is still useful for us.

So the threshold oracle lets us use noisy sum answers (allowing the addition of noise and differential privacy), but in fact
it also addresses the second challenge of reducing the query complexity of
the learning algorithm. This is described next.

\myparagraph{Savings in Query Complexity via Subsampling.}
The remaining challenge is that the threshold oracle can be invoked only (at most)
$n$ times before we exceed our ``privacy budget''. This is problematic, because the
query complexity of the underlying learning algorithm may well depend on~$n$, since $f^D$ is
a sum of $n$ predicates.
To reduce the number of oracle queries that need to be made, we observe that
the sum of $n$ predicates that we are trying to learn can actually be
approximated by a sum of fewer predicates. In fact, there exists a sum
$f^{D'}$ of $n' = O(\log |\Q|/\alpha^2)$ predicates from $\F$ that is
$\alpha/100$-close to $f^D$ on all inputs in $\cQ,$ i.e.
$|f^D(q)-f^{D'}(q)|\le\alpha/100$ for all $q\in\Q.$
(The proof is by a subsampling argument, as
in \cite{BlumLR08}; see Section~\ref{sec:oracle}.)
We will aim to learn this ``smaller'' sum. The hope is that the query
complexity for learning $f^{D'}$ may be considerably smaller, namely scaling
with $n'$ rather than $n$. Notice, however, that learning a threshold of
$f^{D'}$ requires a threshold oracle to $f^{D'}$, rather than the threshold
oracle we have, which is to $f^D$. Our goal, then, is to use the threshold
oracle to $f^D$ to simulate a threshold oracle to $f^{D'}$. This will give us
``the best of both worlds'': we can make (roughly) $O(n)$ oracle queries thus
preserving differential privacy, while using a learning algorithm that
is allowed to have query complexity superlinear in~$n'.$

The key observation showing that this is indeed possible is that 
the threshold oracle $\TO$ already ``avoids'' low-margin
queries where $f^D_t$ and $f^{D'}_t$ might disagree!  Whenever the
threshold oracle $\TO$ (w.r.t.~$D$) answers $l\ne\bot$ on a query $q,$, we
must have $|f^D(q) - t| \geq \alpha/20 - N > \alpha/100$, and thus
$f^D_t(q)=f^{D'}_t(q).$ Moreover, it is still the case that $\TO$ only
answers $\perp$ on queries $q$ that have low margins w.r.t $f^{D'}_t$.
This means that, as above, we can run $\Le$ using $\TO$ (w.r.t. $D$)
in order to learn $f^{D'}$. The query complexity depends on $n'$ and is
therefore independent of $n.$
At the same time, we continue to answer all
queries using the threshold oracle with respect to $f^D$ so that our privacy
budget remains on the order~$|D|=n$.
Denoting the query complexity of the learning algorithm by
$b(n')$ we only need that $n\gg b(n').$ This allows us to use learning
algorithms that have $b(n')\gg n'$ as is usually the case.

\myparagraph{Sampling from the conditional distribution.}
In the exposition above we glossed over one technical detail, which is that
the learning algorithm requires sampling (or distribution restricted) access
to the distribution $G'$ over queries $q$ on which  $\TO$ does not return
$\bot$, whereas the data release algorithm we are trying to build only has
access to the original distribution $G$. We reconcile this disparity as
follows.

For a threshold $t$, let $\zeta_t$ denote the probability that the oracle~$\TO$
does not return~$\bot$ when given a random $q \sim G$ and the threshold~$t.$
There are two cases depending on $\zeta_t$:

\begin{description}

\item[$\zeta_t<\gamma$:] This means that the threshold $t$ is such that
with probability $1-\gamma$ a random sample $q\sim G$ has low margin with respect to $f^D$
and $t$.  In this
case, by simply outputting the constant-$t$ function as our approximation for
$f^D$, we get a hypothesis that has accuracy $\alpha/3$ with probability
$1-\gamma$ over random $q\sim G.$

\item[$\zeta_t\ge\gamma$:] In this case, the conditional distribution $G'$ induced by the threshold oracle is $1/\gamma$-smooth w.r.t. $G$. In particular, $G'$ is contained in the smooth extension
$\GQ'$ for which the learning algorithm is guaranteed to work (by the conditions of Theorem~\ref{thm:privacy-via-learning}).  This means that it we can sample from $G'$ using rejection sampling to $G$. It suffices to oversample by a factor of $O(1/\gamma)$ to make sure that we
receive enough examples that are not rejected by the threshold oracle.
\end{description}
Finally using a reasonably accurate estimate of $\zeta,$ we can also implement
the distribution restricted approximate oracle access that may be required by the learning
algorithm. We omit the details from this informal overview.

\section{Proof of Theorem~\ref{thm:main}}
\label{sec:mainproof}

In this section, we give a formal proof of Theorem~\ref{thm:main}. We
formalize and analyze the threshold oracle first. Then we proceed to our main
reduction.

\subsection{Threshold access and subsampling}
\label{sec:oracle}

We begin by describing the threshold oracle that we use to access the function
$f^D$ throughout our reduction; it is presented in Figure~\ref{fig:oracle}.  The oracle has two purposes.
One is to ensure differential privacy by adding noise every time we
access $f^D.$ The other purpose is to ``filter out'' queries that are too
close to the given threshold. This will enable us to argue that the threshold
oracle for $f_t^D$ agrees with the function $f_t^{D'}$ where $D'$ is a small
subsample of~$D$.

\begin{figure}[h]
\begin{boxedminipage}{0.99\textwidth}
\noindent {\bf Input:}
data set $D$ of size $n,$
tolerance $\alpha>0,$ query bound
$b\in\mathbb{N}.$

{\bf Threshold Oracle} $\mathcal{TO}(D,\alpha,b)$:
\begin{itm}
\item When invoked on the $j$-th query $(q,t) \in \Q\times[0,1),$
do the following:
\begin{itm}
\item If $j>b,$ output $\bot$ and terminate.
\item If $(q,t')$ has not been asked before for any threshold $t',$
sample a fresh Laplacian variable $N_q\sim\Lap(b/\eps n)$ and put
$A_q=f^D(q)+N_q.$ Otherwise reuse the previously created value $A_q.$
\item
\[
\text{Output~}
\begin{cases}
0 & \text{if $A_q \leq t-{2\alpha}/3$,}\\
1 & \text{if $A_q \geq t+{2\alpha}/3$,}\\
\bot & \text{otherwise.}
\end{cases}
\]
\end{itm}
\end{itm}
\end{boxedminipage}
\caption{Threshold oracle for $f^D.$ This threshold oracle is the only way in
which the data release algorithm ever interacts with the data set $D.$
Its purpose is to ensure privacy
and to reject queries that are too close to a given threshold.}
\label{fig:oracle}
\end{figure}

Throughout the remainder of this section we fix all input parameters to our
oracle, i.e. the data set $D$ and the values $b,\alpha>0.$ We let $\beta>0$
denote the desired error probability of our algorithm.

\begin{lemma}
\label{lem:oracle-privacy}
Call two queries $(q,t),(q',t')$ \emph{distinct} if $q\ne q'.$ Then,
the threshold oracle $\TO(D,\alpha,b)$
answers any sequence of $b$ distinct adaptive queries to $f^D$ with
$\epsilon$-differential privacy.
\end{lemma}

\begin{proof}
This follows directly from the guarantees of the Laplacian mechanism as shown
in~\cite{DworkMNS06}.
\end{proof}

Our goal is to use the threshold oracle for $f_t^D$ to correctly answer queries
to the function $f_t^{D'}$ where $D'$ is a smaller (sub-sampled) database that gives ``close'' answers to $D$ on all queries $q \in \Q$.
The next lemma shows that there always exists such a smaller database.
\begin{lemma}\label{lem:subsample}
For any $\alpha \geq 0$, there is a database $D'$ of size
\begin{equation}\label{eq:subsample-size}
|D'|\le \frac{10\log|\cQ|}{\alpha^2}
\end{equation}
such that
\[
\max_{q\in \cQ} \left|f^D(q)-f^{D'}(q)\right|<\alpha\mper
\]
\end{lemma}
\begin{proof}
The existence of~$D'$ follows from a subsampling argument as shown
in~\cite{BlumLR08}.
\end{proof}
The next lemma states the two main properties of the threshold oracle
that we need. To state them more succinctly, let us denote by
\[
Q(t,\alpha)=\{q\in \Q\colon |f^D(q)-t|\ge\alpha\}
\]
the set of elements in $\Q$ that are $\alpha$-far from the threshold~$t.$

\begin{lemma}[Agreement]\label{lem:agree}
Suppose $D$ satisfies
\begin{equation}
\label{eq:agree-db}
|D|\ge\frac{30b\cdot\log(\nfrac b\beta)}{\epsilon\alpha},
\end{equation}
Then, there is a data set $D'$ of size $|D'|\le90\cdot\alpha^{-2}\log|\cQ|$ and an
event~$\Gamma$ (only depending on the choice of the Laplacian variables)
such that $\Gamma$ has probability $1-\beta$ and
if $\Gamma$ occurs, then $\TO(D,\alpha,b)$ has the following guarantee:
whenever $\TO(D,\alpha,b)$ outputs $l$ on one of the queries $(q,t)$ in the
sequence, then
\begin{enumerate}
\item\label{prop:agree-1} if $l\ne\bot$ then
$l=f_t^{D'}(q)=f_t^D(q)\mcom$ and
\item\label{prop:agree-2} if $l=\bot$ then $q\not\in Q(t,\alpha)\mper$
\end{enumerate}
\end{lemma}

\begin{proof}
Let $D'$ be the data set given by Lemma~\ref{lem:subsample} with its
``$\alpha$'' value set to $\alpha/3$ so that
\[
\left|f^D(q)-f^{D'}(q)\right|<\alpha/3
\]
for every input $q\in\Q.$

The event $\Gamma$ is defined as the event that
every Laplacian variable~$N_q$ sampled
by the oracle has magnitude $|N_q|<\nfrac\alpha3.$
%
Under the given assumption on~$|D|$ in~\ref{eq:agree-db}
and using basic tail bounds for the Laplacian distribution,
this happens with probability $1-\beta.$

Assuming $\Gamma$ occurs,
the following two statements hold:
\begin{enumerate}
\item Whenever the oracle outputs $l\ne\bot$ on a query $(q,t),$
then we must have
either $f^D(q)+N_q-t\ge2\alpha/3$
(and thus both $f^D(q)> t+\nfrac\alpha3$ and
$f^{D'}(q)> t$)
or else $f^D(q)+N_q-t \leq -2\alpha/3$
(and thus both $f^D(q)< t-\nfrac\alpha3$ and
$f^{D'}(q)< t$). This proves the first claim of the lemma.
\item Whenever $q\in Q(t,\alpha),$ then $|f^D(q)+N_q-t|\ge
2\alpha/3$, and therefore the oracle does not output $\bot.$ This
proves the second claim of the lemma.
\end{enumerate} \end{proof}

\subsection{Privacy-preserving reduction}
\label{sec:reduction}

In this section we describe how to convert a non-private learning algorithm
for threshold functions of the form $f_t^D$
to a privacy-preserving learning algorithm for functions of the form $f^D.$
The reduction is presented in Figure~\ref{fig:reduction}. We call the
algorithm \PrivLearn.

\begin{figure}[h]
\begin{boxedminipage}{0.99\textwidth}
\noindent{\bf Input:}
Distribution $G\in\GQ,$ data set $D$ of size $n,$ accuracy parameters
$\alpha,\beta,\gamma>0;$
learning algorithm $\Le$ for thresholds over $(\Q,\GQ,\cF)$ as in
Theorem~\ref{thm:main} requiring $b(n',\gamma',\beta')$ labeled examples
and approximate restricted evaluation access to the target function.

\noindent{\bf Parameters:} See~(\ref{eq:params}) and~(\ref{eq:params2}).

\noindent{\bf Algorithm \PrivLearn for privately learning $f^D$:}
\begin{enumerate}
\item Let $\TO$ denote an instantiation of $\TO(D,\nfrac\alpha7,b_\total).$
\item Sample $b_\iter$ points $\{q_j\}_{1\le j\le b_\iter}$ from~$G.$
\item\label{step:loop} For each iteration $i\in\{1,\dots,k\}:$
\begin{enum}
\item
Let  $t_i=\nfrac i{k+1}.$
\item
For each $q_j,$ $j\in[b_\iter]$ send the query $(q_j,t_i)$ to $\TO$ and
let~$l_j$ denote the answer. Let $B_i=\{j\colon l_j\ne\bot\}.$
\item
If $\frac{|B_i|}{b_\iter}< \frac\gamma2,$ output the constant
$t_i$ function as hypothesis~$h$ and terminate the algorithm.
\item
Run the learning algorithm $\Le(n',t_i,\gamma',\beta')$ on the labeled examples
$\{(q_j,l_j)\}_{j\in B_i},$ answering evaluation queries from $\Le$ as follows:
\begin{itm}
\item
Given a query $q$ posed by~$\Le,$ let $l$ be the answer of $\TO$ on $(q,t_i).$
\item If $l=\bot,$
then output $(0,\bot).$ Otherwise, output
$(G[q]\cdot \frac{b_\iter}{|B_i|},l)\mper$
\end{itm}
\item
Let $h_i$ denote the resulting hypothesis.
\end{enum}
\item Having obtained hypotheses $h_1,\dots,h_k,$ the final
hypothesis $h$ is defined as follows:  $h(q)$ equals the smallest $i\in[k]$
such that $h_i(q)=1$ and $h_{i-1}(q)=0$ (we take $h_{0}(q)=0$
and $h_{k+1}(q)=1$).
\end{enumerate}
\end{boxedminipage}
\caption{Reduction from private data release to learning thresholds
(non-privately).}
\label{fig:reduction}
\end{figure}

\myparagraph{Setting of parameters.}
In the description of \PrivLearn we use the following setting of parameters:
\begin{equation}\label{eq:params}
n'=\frac{4410\cdot\log|\Q|}{\alpha^2}
\qquad k = \left\lceil \frac3\alpha \right\rceil
\qquad\gamma'=\frac\gamma k
\qquad\beta'=\frac\beta {6k}
\end{equation}
\begin{equation}\label{eq:params2}
b_\base = b(n',\gamma',\beta')
\qquad b_\iter = \frac{100 b_\base\cdot \log(\nfrac 1{\beta'})}\gamma
\qquad b_\total = 2k\cdot b_\iter
\end{equation}

\myparagraph{Analysis of the reduction.}
Throughout the analysis of the algorithm we keep all input parameters fixed so
as to satisfy the assumptions of Theorem~\ref{thm:main}.
Specifically we will need
\begin{equation}\label{eq:dbsize}
|D|\ge
\frac{210\cdot b_\total\cdot\log(10 b_\total/\beta)}{\epsilon\alpha}\mper
\end{equation}
We have made no attempt to optimize various constants throughout.
\begin{lemma}[Privacy]
\label{lem:privacy}
Algorithm \PrivLearn satisfies $\epsilon$-differential privacy.
\end{lemma}
\begin{proof}
In each iteration of the loop in Step~\ref{step:loop}
the algorithm makes at most~$2b_\iter$ queries to~$\TO$ (there are $b_\iter$ calls made on the samples and at most $b_\base\le b_\iter$
evaluation queries).
But note that $\TO$ is instantiated with a query bound of
$b_\total =2kb_\iter.$ Hence, it follows from
Lemma~\ref{lem:oracle-privacy} that $\TO$ satisfies $\epsilon$-differential
privacy. Since~$\TO$ is the only way in which \PrivLearn ever interacts with
the data set, \PrivLearn satisfies $\epsilon$-differential privacy.
\end{proof}
We now prove that the hypothesis produced by the algorithm is indeed
accurate, as formalized by the following lemma.
\begin{lemma}[Accuracy]\label{lem:accuracy}
With overall probability $1-\beta,$
the hypothesis $h$ returned by \PrivLearn satisfies
\begin{equation}\label{eq:accuracy}
\Pr_{q\sim G}\left\{\left|h(q)-f^{D}(q)\right|\le\alpha\right\}
\ge1-\gamma\mper
\end{equation}
\end{lemma}

\begin{proof}
We consider three possible cases:
\begin{enumerate}
\item The
first case is that there exists a $t \in \{t_1,\dots,t_k\}$ such that
distribution $G$ has at least $1-\gamma/10$ of its mass on points that are
$\alpha$-close to $t$.  In this case a Chernoff bound and the choice of
$b_\iter\gg b_\base$ imply that with probability $1-\beta$ the algorithm
terminates prematurely and the resulting hypothesis
satisfies~(\ref{eq:accuracy}).

\item
In the second case, there exists a $t \in \{t_1,\dots,t_k\}$ such that
the probability mass $G$ puts on points that are $\alpha$-close to $t$ is
between $1-\gamma$ and $1-\gamma/10.$  In this case if the algorithm
terminates prematurely then~(\ref{eq:accuracy}) is satisfied; below we analyze
what happens assuming the algorithm does not terminate prematurely.

\item
In the third case every $t \in \{t_1,\dots,t_k\}$
is such that $G$ puts less than $1-\gamma$ of its mass on points
$\alpha$-close to $t.$  In this third case if the algorithm terminates
prematurely then~(\ref{eq:accuracy}) will not hold; however, our choice of $b_\iter$ implies
that in this third case the algorithm terminates prematurely with probability
at most $1-\beta$.  As in the second case, below we will analyze what happens
assuming the algorithm does not terminate prematurely.
\end{enumerate}
Thus in the remainder of the argument we may assume without loss of generality
that the algorithm does not terminate prematurely, i.e. it produces a full
sequence of hypotheses $h_1,\dots,h_k.$ Furthermore, we can assume that the
distribution $G$ places at most $1-\gamma/10$ fraction of its weight near any
particular threshold~$t_i.$ This leads to the following claim,
showing that in all iterations, the number of labeled examples in $B_i$ is
large enough to run the learning algorithm.
\begin{claim}\label{cla:enough}
$\Pr\Set{\forall i\colon |B_i|\ge b_\base}\ge 1-\frac\beta3\mper$
\end{claim}
\begin{proof}
By our assumption, the probability that a sample $q\sim G$ is rejected at
step~$t$ of \PrivLearn is at most~$\nfrac\gamma{10}.$ By the choice
of~$b_\iter$ it follows that $|B_i|\ge b_\base$ with probability $1-\beta/k.$
Taking a union bound over all thresholds~$t$ completes the proof.
\end{proof}

The proof strategy from here on
is to first analyze the algorithm on the conditional
distribution that is induced by the threshold oracle. We will then pass from
this conditional distribution to the actual distribution that we are
interested in, namely, $G.$

We chose $|D|$ large enough so that we can apply
Lemma~\ref{lem:agree} to~$\TO$ with the ``$\alpha$''-setting of
Lemma~\ref{lem:agree} set to~$\nfrac\alpha7.$
Let $D'$ be the data set and $\Gamma$ be the event
given in the conclusion of Lemma~\ref{lem:agree} applied to~$\TO$. (Note that
$n'=|D'|\le 7^2\cdot90\alpha^{-2}\log|\Q|$ as stated above.)

By the choice of our parameters, we have
\begin{equation}\label{eq:Gamma}
\Pr\Set{\Gamma} \ge1-\frac\beta{3}\mper
\end{equation}
Here the probability is computed only over the internal randomness of the
threshold oracle $\TO$ which we denote by~$R.$
Fix the randomness~$R$ of $\TO$ such that $R\in\Gamma.$
For the sake of analysis, we can think of the randomness of the oracle as a
collection of independent random variables $(N_q)_{q\in\Q}$ (where $N_q$ is
used to answer all queries of the form $(q,t')$). In particular, the behavior
of the oracle would not change if we were to sample all variables
$(N_q)_{q\in\Q}$ up front. When we fix $R$ we thus mean that we fix $N_q$ for
all $q\in\Q.$

We may therefore assume for the remainder of the analysis that $\TO$ satisfies
properties~(\ref{prop:agree-1}) and~(\ref{prop:agree-2}) of
Lemma~\ref{lem:agree}.

Let us denote by $Q_i\subseteq\Q$
the set of examples for which $\TO$ would not
answer $\bot$ in Step~\ref{step:loop} at the $i$-th iteration of the algorithm.
Note that
this is a well-defined set since we fixed the randomness of the oracle.
Denote by $G_i$ the distribution $G$ conditioned on $Q_i.$
Further, let
$Z_i = \Pr_{q\sim G}\Set{q\in Q_i}\mper$
Observe that
\begin{equation}\label{eq:Gt}
G_i[q]=\begin{cases}
G[q]/Z_i & q\in Q_i\\
0& \text{ o.w. }
\end{cases}
\mper
\end{equation}
The next lemma shows that \PrivLearn answers evaluation queries with the
desired multiplicative precision.
\begin{lemma}
\label{lem:Zt}
With probability $1-\nfrac\beta{6k}$ (over the randomness of \PrivLearn), we
have
\begin{equation}\label{eq:Zt}
\frac {Z_i}{3} \le \frac {|B_i|}{b_\iter} \le 3{Z_i}\mper
\end{equation}
\end{lemma}
\begin{proof}
The lemma follows from a Chernoff bound with the fact that we chose
$b_\iter\gg b_\base$.
\end{proof}
Assuming that~(\ref{eq:Zt}) holds, we can argue that the learning algorithm in
step~$t$ produces a ``good'' hypothesis as expressed in the next lemma.
\begin{lemma}
\label{lem:hi}
Let $t\in\{t_1,\dots,t_k\}.$
Conditioned on~(\ref{eq:Zt}), we have that with
probability $1-\nfrac\beta{6k}$
(over the internal randomness of the learning algorithm invoked at step~$i$)
the hypothesis $h_i$ satisfies
\[
\Pr_{q\sim G_i}\left\{h_i(q)=f^{D}_{t_i}(q)\right\}\ge1-\frac\gamma k\mper
\]
\end{lemma}
\begin{proof}
This follows directly from the guarantee of the learning algorithm~$\Le$
once we argue that (with the claimed probability):
\begin{enumerate}
\item Each example $q$ is sampled from $G_i$ and labeled correctly by
$f^{D'}_{t_i}(q)$ and $f^{D'}_{t_i}(q)=f^{D}_{t_i}(q).$
\item All evaluation queries asked by the learning algorithm are answered
with the multiplicative error allowed in Definition~\ref{def:distribution-oracle-access}.
\item The algorithm received sufficiently many, i.e., $b_\base$, labeled
examples.
\end{enumerate}
The first claim follows from the definition of $G_i$, since we can sample
from $G_i$ by sampling from $G$ and rejecting if the oracle~$\TO$ returns
$\bot.$ Since $\Gamma$ is assumed to hold,  we can invoke
property~(\ref{prop:agree-1}) of Lemma~\ref{lem:agree} to conclude that
whenever the oracle does not return
$\bot,$ then its answer agrees with $f^{D'}_{t_i}(q)$ and moreover
$f^{D'}_{t_i}(q)=f^{D}_{t_i}(q).$

To see the second claim, consider an evaluation query $q.$ We
consider two cases. The first case is
where the threshold oracle returns $\bot$ and \PrivLearn outputs~$(0,\bot).$ Note that in this
case indeed $G_i$ puts $0$ weight on the query $q.$ In the second case
\PrivLearn outputs $(G[q]\cdot b_\iter/|B_i|,l).$
By~(\ref{eq:Gt}) and since we assumed $\Gamma$
holds, the output satisfies the desired multiplicative bound.

The third claim is a direct consequence of Claim~\ref{cla:enough}.
\end{proof}

We conclude from the above that
with probability $1-\beta/3$
(over the combined randomness of \PrivLearn and of the learning algorithm),
simultaneously for all $i\in[k]$ we have
\begin{equation}\label{eq:hi}
\Pr_{q\sim G}\Set{h_i(q)\ne f^{D}_{t_i}(q)\Mid Q_i}
=
\Pr_{q\sim G_i}\Set{h_i(q)\ne f^{D}_{t_i}(q)}
\le \frac\gamma{k}\mper
\end{equation}
This follows from a union bound over all $k$ applications of
Lemma~\ref{lem:Zt} and Lemma~\ref{lem:hi}.

We can now complete the proof of Lemma~\ref{lem:accuracy}. That is, we will
show that assuming~(\ref{eq:hi}) the hypothesis $h$ satisfies
\[
\Pr_{q\sim G}\left\{\left|h(q)-f^{D}(q)\right|\le\alpha\right\}
\ge1-\gamma\mper
\]
Note that
\begin{enumerate}
\item (\ref{eq:hi}) occurs with probability $1-\nfrac\beta3,$
\item our assumption on the threshold oracle, i.e., $R\in\Gamma$ also occurs with
probability $1-\nfrac\beta3$ (over the randomness of the oracle)
\item the event in Claim~\ref{cla:enough} holds with
probability~$1-\nfrac\beta3.$
\end{enumerate}
Hence all three events occur simultaneously with probability~$1-\beta$
which is what we claimed. We proceed by assuming that all three events
occurred.
\newcommand\Err{\mathrm{Err}}
In the following, let
\[
\Err_i=\{q\in\Q\colon h_i(q)\ne f^D_{t_i}(q)\}
\]
denote the set of points on which $h_i$ errs.
We will need the following claim.

\begin{claim}\label{cla:wtf}
Let $q\in\Q.$ Then,
\[
\left|h(q)-f^{D}(q)\right|> \alpha
\qquad\Longrightarrow\qquad
q\in\bigcup_{i\in[k]} \Err_i\cap Q_i\mper
\]
\end{claim}
\begin{proof}
Arguing in the contrapositive, suppose
$q\not\in \bigcup_{i\in[k]} \Err_i\cap Q_i.$
This means that for all $i\in[k]$ we have that either
$q\not\in \Err_i$ or $q\not\in Q_i.$

However, we claim that there can be at most one $i\in[k]$ such that $q\not\in
Q_i$ meaning that $q$ is rejected at step $i.$
This follows from property~(\ref{prop:agree-2}) of
Lemma~\ref{lem:agree} which asserts that if $q\not\in Q_i,$ then we must have
$|f^D(q)-t_i|<\alpha/7$, and the fact that any two thresholds differ by at least
$\alpha/3.$

Hence, under the assumption above it must be the case that $q\not\in \Err_i$ for
all but at most one $i\in[k].$ This means that all but one hypothesis $h_i$
correctly classify $q.$ Since the thresholds are spaced $\alpha/3$ apart,
this means the hypothesis $h$ has error at most $2\alpha/3\le\alpha$
on~$q$.
\end{proof}

With the previous claim, we can finish the proof. Indeed,
\begin{align*}
\Pr_{q\sim G}\Set{ \left|h(q) - f^D(q)\right| > \alpha }
& \le \Pr_{q\sim G}\Set{ \bigcup_{i\in[k]} \Err_i \cap Q_i }
\qquad \tag{using~Claim~\ref{cla:wtf}}\\
& \le \sum_{i=1}^k \Pr_{q\sim G}\Set{ \Err_i \cap Q_i }
\qquad \tag{union bound}\\
& = \sum_{i=1}^k \Pr_{q\sim G}\Set{ q \in \Err_i \mid Q_i }\Pr_{q\sim G}\Set{Q_i}\\
& \le \sum_{i=1}^k \Pr_{q\sim G}\Set{ q \in \Err_i \mid Q_i }\\
& \le k \cdot \frac\gamma k
\qquad \tag{using~(\ref{eq:hi})}\\
&= \gamma \mper
\end{align*}
This concludes the proof of Lemma~\ref{lem:accuracy}
\end{proof}

Lemma~\ref{lem:privacy} (Privacy) together with Lemma~\ref{lem:accuracy}
(Accuracy) conclude the proof of out main theorem, Theorem~\ref{thm:main}.

\subsection{Quantitative Improvements without Membership Queries}

Here we show how to shave off a factor of $1/\alpha$ in the requirement on the
data set size~$n$ in Theorem~\ref{thm:main}. This is possible if
the learning algorithm uses only sampling access to labeled examples.

\begin{theorem} \label{thm:main-random}
Let $\U$ be a data universe, $\Q$ a set of query descriptions, $\GQ$ a set of
distributions over $\Q$, and $P\colon\Q \times \U \rightarrow \zo$ a
predicate.

Then, there is an $\eps$-differentially private
$(\alpha,\beta,\gamma)$-accurate data-release algorithm provided that
there is an algorithm $\Le$ that ($\gamma$,$\beta$)-learns thresholds
over $(\Q,\GQ',\{p_u\colon u\in\U\})$ using $b(n,\gamma,\beta)$
random examples; and
we have
\[
n\ge
\frac{C\cdot
b(n',\gamma',\beta')\cdot\log\left(\frac{b(n',\gamma',\beta')}{\alpha\gamma\beta}
\right)
\cdot \log(1/\beta')}{\epsilon\alpha\gamma}\mcom
\]
where $n'=\Theta(\log|\cQ|/\alpha^2),$
$\beta'=\Theta(\beta\alpha),$
$\gamma'=\Theta(\gamma\alpha)$ and $C>0$ is a sufficiently large constant.
If ${\cal L}$ runs in time $t(n,\gamma,\beta)$ then the data release algorithm runs in time
$\poly(t(n',\gamma',\beta'),n,1/\alpha,\log(1/\beta),1/\gamma)$.

\end{theorem}

\begin{proof}
The proof of this theorem is identical to that of Theorem~\ref{thm:main}
except that we put $b_\total=2b_\iter$ rather than $2kb_\iter.$ It is easy to
check that the algorithm indeed makes only $b_\total$ \emph{distinct queries}
(in the sense of Lemma~\ref{lem:oracle-privacy}) to the threshold
oracle so that privacy remains ensured.
The correctness argument is identical.
\end{proof}

\section{First Application: Data Release for Conjunctions}
\label{sec:conj}
\label{sec:first-app}

With Theorems~\ref{thm:main} and~\ref{thm:main-random} in hand, we can obtain
new data release algorithms ``automatically'' from learning algorithms that
satisfy the properties required by the theorem.  In this section we present
such data release algorithms for conjunction counting queries using learning
algorithms (which require only random examples and work under any
distribution) based on polynomial threshold functions.

Throughout this section we fix the query class under consideration to be monotone conjunctions, i.e. we take
$\U=\Q=\{0,1\}^d$ and
$P(q,u)= 1-\bigvee_{i\colon u_i=0} q_i$.

The learning results given later in this section, together with
Theorem~\ref{thm:main-random}, immediately yield:

\begin{theorem}[Releasing conjunction counting queries]
\label{thm:conjunction}

\begin{enumerate}

\item  There is an $\eps$-differentially private algorithm for releasing the class of monotone Boolean conjunction queries over $\GQ=\{$all probability distributions over $\Q\}$ which is
    $(\alpha,\beta,\gamma)$-accurate and has runtime $\poly(n)$
for databases of size $n$ provided that
\[
n \geq
d^{O\left(d^{\nfrac13}\log\left(\frac d{\alpha}\right)^{\nfrac23}\right)}\cdot
\tilde O\left(
\frac{\log\left(\nfrac1\beta\right)^3}
{\epsilon\alpha\gamma^2} \right) \mper
\]
\item
There is an $\eps$-differentially private algorithm for releasing the class of monotone Boolean conjunction queries over $\GQ_k=\{$all probability distributions over $\Q$ supported on
$B_k = \{q \in \Q: q_1 + \cdots + q_d \leq k\}\}$ which is
    $(\alpha,\beta,\gamma)$-accurate and has runtime $\poly(n)$ for databases of
size $n$ provided that
\[
n \geq
d^{O\left(\sqrt{k\log\left(\frac{k\log d}{\alpha}\right)}\right)}\cdot
\tilde O\left(
\frac{\log\left(\nfrac1\beta\right)^3}
{\epsilon\alpha\gamma^2}\right)\mper
\]
\end{enumerate}
\end{theorem}

These algorithms are distribution-free, and so we can apply the boosting
machinery of \cite{DworkRV10} to get accurate answers to {\em all} of the
$k$-way conjunctions with similar database size bounds. See the discussion and
Corollary \ref{cor:boosted-conjunctions} in the introduction.

In Section~\ref{sec:ptf} we establish structural results showing that certain
types of thresholded real-valued functions can be expressed as low-degree
polynomial threshold functions.  In Section~\ref{sec:ptflearn} we state some
learning results (for learning under arbitrary distributions) that follow from
these representational results.  Theorem~\ref{thm:conjunction} above follows immediately from
combining the learning results of Section~\ref{sec:ptflearn} with Theorem~\ref{thm:main-random}.

\subsection{Polynomial threshold function representations}
\label{sec:ptf}

\begin{definition}
Let $X \subseteq \Q = \{0,1\}^d$ and let $f$ be a Boolean function $f: X \to \bits$.
We say that $f$ has a \emph{polynomial threshold function (PTF) of degree $a$
over $X$} if there is a real polynomial $A(q_1,\dots,q_d)$ of degree $a$ such that
\[
f(q) = \sign(A(q)) \quad \text{for all~}q \in X
\]
where the $\sign$ function is $\sign(z) =1$ if $z \geq 0$, $\sign(z)=0$ if $z < 0.$
\end{definition}
Note that the polynomial $A$ may be assumed without loss of generality to be
multilinear since $X$ is a subset of $\{0,1\}^d.$

\subsubsection{Low-degree PTFs over sparse inputs}

Let $B_k \subset \{0,1\}^d$ denote the collection of all points with Hamming
weight at most $k$, i.e. $B_k = \{q \in \{0,1\}^d: q_1 + \cdots + q_d \leq
k\}.$ The main result of this subsection is a proof that for any $t \in [0,1]$
the function $f^D_t$ has a low-degree polynomial threshold function over~$B_k.$

\begin{lemma}\label{lem:thr-rep}
Fix $t \in [0,1].$  For any database $D$ of size $n$, the function $f^D_t$ has a polynomial threshold function
of degree $O\left(\sqrt{k \log n}\right)$ over the domain $B_k.$
\end{lemma}

To prove Lemma~\ref{lem:thr-rep} we will use the following claim:

\begin{claim}\label{claim:lee}
Fix $k > 0$ to be a positive integer and $\eps > 0.$  There is a univariate
polynomial $s$ of degree $O\left(\sqrt{k \log(\nfrac1\eps)}\right)$ which is such that
\begin{enumerate}
\item $s(k)=1$; and
\item $|s(j)| \leq \eps$ for all integers $0 \leq j \leq k-1.$
\end{enumerate}
\end{claim}
\begin{proof}
This claim was proved by Buhrman et al. \cite{BuhrmanCWZ99}, who gave a quantum algorithm which
implies the existence of the claimed polynomial (see also Section~1.2 of \cite{Sherstov09}).
Here we give a self-contained construction of
a polynomial $s$ with the claimed properties that satisfies the slightly weaker degree bound
$\deg(s) = O(\sqrt{k}\log(1/\eps)).$
We will
use the univariate Chebyshev polynomial $C_r$ of degree $r=\lceil \sqrt{k} \rceil.$
Consider the polynomial
\begin{equation}\label{eq:poly-rep}
s(j) = \left({\frac {C_r\left( \frac jk\left(1+\frac1k\right)\right)}{C_r\left(1 + {\frac 1 k}\right)}}\right)^{\lceil
\log(1/\eps)\rceil}.
\end{equation}
It is clear that if $j=k$ then $s(j)=1$ as desired, so suppose
that $j$ is an integer $0 \leq j \leq k-1$.  This implies that $(j/k)(1 + 1/k) <1$.
Now well-known properties of the Chebychev polynomial (see e.g. \cite{Cheney66})
imply that $|C_r((j/k)(1+1/k))| \leq 1$ and $C_r(1+1/k) \geq 2.$  This gives the $O(\sqrt{k}\log(1/\eps))$ degree bound.
\end{proof}

Recall that the predicate function for a data item $u\in\bits^d$ is denoted by
\[
p_u(q)= 1-\bigvee_{i\colon u_i=0} q_i\mper
\]
As an easy corollary of Claim~\ref{claim:lee} we get:
\begin{corollary} \label{cor:poly-rep}
Fix $\eps > 0.$
For every $u \in \zo^d$, there is a $d$-variable
polynomial $A_u$ of degree $O\left(\sqrt{k \log(\nfrac1\eps)}\right)$
which is such that for every $q \in B_k$,
\begin{enumerate}
\item If $p_u(q)=1$ then $A_u(q)=1$;
\item If $p_u(q)=0$ then $|A_u(q)|\le \eps.$
\end{enumerate}
\end{corollary}

\begin{proof}
Consider the linear function
$L(q) = k - \sum_{i\colon u_i=0} q_i.$
For $q \in B_k$ we have that $L(q)$ is an integer in $\{0,\dots,k\}$, and we have $L(q)=k$ if and only if
$p_u(q)=1.$   The desired polynomial is $A_u(q)=s(L(q)).$
\end{proof}

\begin{proof}[Proof of Lemma~\ref{lem:thr-rep}]
Consider the polynomial
\[
A(q) = \sum_{u\in D} A_u(q)
\]
where for each data item $u$, $r_u$ is the polynomial from Corollary~\ref{cor:poly-rep} with
its ``$\eps$'' parameter set to $\eps=1/(3n).$  We will show that $A(q) - (\lceil t n \rceil - 1/2)$ is the desired polynomial which gives a PTF for $f^D_t$ over $B_k.$

First, consider any fixed $q \in B_k$ for which $f^D_t(q) =1$.  Such a $q$ must satisfy
$f^D(q) = j/n \geq t$ for some integer $j$, and hence $j \geq \lceil t n \rceil.$ Corollary~\ref{cor:poly-rep}
now gives that $A(q) \geq \lceil t n \rceil - 1/3.$

Next, consider any fixed $q \in B_k$ for which $f^D_t(q)=0.$  Such a $q$ must satisfy
$f^D(q) = j/n < t$ for some integer $j$, and hence $j \leq \lceil t n \rceil - 1.$
Corollary~\ref{cor:poly-rep} now gives that $A(q) \leq \lceil t n \rceil - 2/3.$  This proves the lemma.
\end{proof}

\subsubsection{Low-degree PTFs over the entire hypercube}

Taking $k=d$ in the previous subsection, the results there imply that $f^D_t$ can be represented by a polynomial
threshold function of degree $O\big(\sqrt{d \log n}\big)$ over the entire Boolean
hypercube $\{0,1\}^d.$ In this section we improve the degree 
to~$O\big(d^{\nfrac13}(\log n)^{2/3}\big).$
This result is very similar to Theorem~8 of \cite{KlivansOS04} (which is
closely based on the main construction and result of \cite{KlivansS04}) but
with a few differences:  first, we use Claim~\ref{claim:lee} to obtain
slightly improved bounds.  Second, we need to use the following notion in
place of the notion of the ``size of a conjunction'' that was used in the
earlier results:

\begin{definition} \label{def:width}
The \emph{width} of a data item $u\in D$ is defined as the number of
coordinates $i$ such that $u_i=0.$ The \emph{width} of $D$ is defined
as the maximum width of any data item $u\in D.$
\end{definition}

We use the following lemma:
\begin{lemma} \label{lem:width}
Fix any $t \in [0,1]$ and suppose that $n$-element database $D$ has width~$w.$ Then
$f_t^D$ has a polynomial threshold function of degree $O\left(\sqrt{w \log
n}\right)$ over the domain $\{0,1\}^d.$
\end{lemma}

\begin{proof}
The proof follows the constructions and arguments of the previous subsection,
but with ``$w$'' in place of ``$k$'' throughout (in particular the linear
function $L(q)$ is now defined to be $L(q)= w-\sum_{i\colon u_i=0}q_i$).
\end{proof}

\begin{lemma} \label{lem:DT}
Fix any value  $r\in\{1,\dots,d\}.$  The function
$f^D_t(q_1,\dots,q_d)$ can be expressed as a decision tree~$T$ in which
\begin{enumerate}
\item each internal node of the tree contains a variable $q_i$;
\item each leaf of $T$ contains a function of the form $f^{D'}_t$ where
$D'\subseteq D$ has width at most $r;$
\item the tree $T$ has rank at most $(2d/r)\ln n + 1.$
\end{enumerate}
\end{lemma}
\begin{proof}[Proof sketch.]
The result follows directly from the proof of Lemma~10 in \cite{KlivansS04},
except that we use the notion of width from Definition~\ref{def:width} in place of the
notion of the size of a conjunction that is used in \cite{KlivansS04}.
To see that this works, observe that since $p_u(q) = 1-\vee_{i\colon u_i=0}q_i,$
fixing $q_i=1$ will fix all predicates $p_u$ with $u_i=0$ to be zero.  Thus the analysis
of \cite{KlivansS04} goes through unchanged, replacing ``terms of $f$ that have
size at least $r$'' with ``data items in $D$ that have width at least $r$'' throughout.
\end{proof}

\begin{lemma}
The function $f^D_t$ can be represented as a polynomial threshold function of
degree $O(d^{\nfrac 13}(\log n)^{\nfrac 23}).$
\end{lemma}

\begin{proof}
The proof is nearly identical to the proof of Theorem~2 in \cite{KlivansS04} but
with a few small changes.  We take $r$ in Lemma~\ref{lem:DT} to be $d^{2/3}(\log n)^{1/3}$
and now apply Lemma~\ref{lem:width} to each width-$r$ database $D'$ at a leaf of the
resulting decision tree.  Arguing precisely as in Theorem~2 of \cite{KlivansS04}
we get that $f^D_t$ has a polynomial threshold function of degree
\[
\max\left\{\tfrac{2d}r\ln n + 1, O\left(\sqrt{r \log n}\right)\right\}
= O\left(\sqrt{r \log n}\right)
= O \left( d^{\nfrac13}(\log n)^{\nfrac23} \right).
\]
\end{proof}

\subsection{Learning thresholds of conjunction queries under arbitrary distributions}
\label{sec:ptflearn}

It is well known that using learning
algorithms based on polynomial-time linear programming, having low-degree PTFs for a class of functions implies efficient PAC learning algorithms for that class under any distribution using random examples only (see
e.g.~\cite{KlivansS04,HellersteinS07}).  Thus the representational results of Section~\ref{sec:ptf} immediately give learning results for the class of
threshold functions over sums of data items.  We state these learning results using the terminology of our reduction below.


\begin{theorem} \label{thm:lbool1}
Let \begin{itemize}
\item $\U$ denote the data universe $\{0,1\}^d$;
\item $\Q$ denote the set of query descriptions $\{0,1\}^d$;
\item $P(q,u)= 1-\bigvee_{i\colon u_i=0} q_i$ denote the monotone conjunction predicate;
\item $\GQ$ denote the set of all probability distributions over $\Q$; and
\item $\GQ_k$ denote the set of all probability distributions over $\Q$ that are supported
 on $B_k = \{q \in \{0,1\}^d: q_1 + \cdots + q_d \leq k\}.$
\end{itemize}
Then

\begin{enumerate}

\item (Learning thresholds of conjunction queries over all inputs)
There is an
algorithm ${\cal L}$ that $(\gamma,\beta)$ learns thresholds over
$(\Q,\GQ,\{p_u: u \in \U\})$ using $b(n,\gamma,\beta)=d^{O(d^{\nfrac13}(\log
n)^{2/3})}\cdot\tilde{O}(1/\gamma)\cdot\log(1/\beta)$ queries to an
approximate distribution-restricted evaluation oracle for the target
$n$-threshold function (in fact ${\cal L}$ only uses sampling access to
labeled examples). The running time of ${\cal L}$ is
$\poly(b(n,\gamma,\beta)).$

\item (Learning thresholds of conjunction queries over sparse
inputs) There is an algorithm ${\cal L}$ that $(\gamma,\beta)$ learns
thresholds over
$(\Q,\GQ_k,\{p_u: u \in \U\})$ using $b(n,\gamma,\beta)=d^{O((k \log n)^{1/2})}\cdot\tilde{O}(1/\gamma)\cdot\log(1/\beta)$ queries to an approximate distribution-restricted evaluation oracle for the target $n$-threshold function (in fact ${\cal L}$ only uses sampling access to labeled examples). The running time of ${\cal L}$ is $\poly(b(n,\gamma,\beta)).$

\end{enumerate}

\end{theorem}

Recall from the discussion at the beginning of Section~\ref{sec:conj} that
these learning results, together with our reduction, give the private data
release results stated at the beginning of the section.

\section{Second Application: Data Release via Fourier-Based Learning}
\label{sec:Fourier}

In this section we present data release algorithms for parity counting queries
and $\AC^0$ counting queries that instantiate our reduction
Theorem~\ref{thm:main} with Fourier-based
algorithms from the computational learning theory literature.  We stress that these algorithms
require the more general reduction Theorem~\ref{thm:main} rather than the simpler version
of Theorem~\ref{thm:intro} because the underlying learning algorithms are not distribution free.
We first give
our results for parity counting queries in Section~\ref{sec:paritylearn} and
then our results for $\AC^0$ counting queries in Section~\ref{sec:AC0}.

\subsection{Parity counting queries using the Harmonic Sieve \cite{Jackson97}}
\label{sec:paritylearn}

In this subsection we fix the query class under consideration to be the class of parity queries, i.e. we take $\U=\{0,1\}^d$ and $\Q=\{0,1\}^d$ and we take $P(q,u)= \sum_{i : u_i=1} q_i \text{~(mod 2)}$ to be the
parity predicate.
Our main result for releasing parity counting queries is:

\begin{theorem}[Releasing parity counting queries]
\label{thm:parity}
There is an $\eps$-differentially private algorithm for releasing the class of
parity queries over the uniform distribution on $\Q$ which is
$(\alpha,\beta,\gamma)$-accurate and has runtime $\poly(n)$ for databases of
size $n$, provided that
\[
n \geq \frac {\poly(d,1/\alpha,1/\gamma,\log(\nfrac1\beta))}{\epsilon}\mper
\]
\end{theorem}

This theorem is an immediate consequence of our main reduction,
Theorem~\ref{thm:main}, and the following learning result:

\begin{theorem} \label{thm:lsieve1}
Let \begin{itemize}
\item $\U$ denote the data universe $\{0,1\}^d$;
\item $\Q$ denote the set of query descriptions $\{0,1\}^d$;
\item $P(q,u) = \sum_{i : u_i=1} q_i \text{~(mod 2)}$ denote the parity predicate; and
\item $\GQ$ contains only the uniform distribution over $\Q$.
\end{itemize}
Then there is an algorithm ${\cal L}$ that $(\gamma,\beta)$ learns thresholds over
$(\Q,\GQ',\{p_u: u \in \U\})$ where $\GQ'$ is the $(2/\gamma)$-smooth extension
of~$\GQ.$  Algorithm
 ${\cal L}$ uses $b(n,\gamma,\beta)=\poly(d,n,1/\gamma) \cdot \log(1/\beta)$
queries to an approximate $G$-restricted evaluation oracle for the target
$n$-threshold function when it is learning with respect to a distribution
$G\in\GQ'.$ The running time of ${\cal L}$ is $\poly(b(n,\gamma,\beta)).$
\end{theorem}

\begin{proof}
The claimed algorithm ${\cal L}$ is essentially Jackson's Harmonic Sieve
algorithm \cite{Jackson97} for learning Majority of Parities; however, a bit
of additional analysis of the algorithm is needed as we now explain.

When Jackson's results on the Harmonic Sieve are expressed in our terminology,
they give Theorem~\ref{thm:lsieve1} exactly as stated above except for one
issue which we now describe. Let $G'$ be any distribution in the
$(2/\gamma)$-smooth extension $\GQ'$ of the uniform distribution.  In
Jackson's analysis, when it is learning a target function $f$ under
distribution $G'$, the Harmonic Sieve is given black-box oracle access to
$f$, sampling access to the distribution $G'$, and access to a
$c$-approximation to an evaluation oracle for $G'$, in the following sense:
there is some fixed constant $c \in [1/3,3]$ such that when the oracle is
queried on $q \in \Q$, it outputs $c \cdot G'[q].$  This is a formally more
powerful type of access to the underlying distribution $G'$ than is allowed
in Theorem~\ref{thm:lsieve1} since Theorem~\ref{thm:lsieve1} only gives ${\cal
L}$ access to an approximate $G'$-restricted evaluation oracle for the target
function (recall Definition~\ref{def:distribution-oracle-access}). To be more
precise, the only difference is that with the Sieve's black-box oracle access
to the target function $f$ it is a priori possible for a learning algorithm to
query $f$ even on points where the distribution $G'$ puts zero probability
mass, whereas such queries are not allowed for ${\cal L}$.  Thus to prove
Theorem~\ref{thm:lsieve1} it suffices to argue that the Harmonic Sieve
algorithm, when it is run under distribution $G'$, never needs to make
queries on points $q \in \Q$ that have $G'[q]=0.$

Fortunately, this is an easy consequence of the way the Harmonic Sieve
algorithm works.  Instead of actually using black-box oracle queries for $f$,
the algorithm actually only ever makes oracle queries to the function
$g(q)=2^d \cdot f(q) \cdot D'[q]$, where $D'$ is a $c$-approximation to an
evaluation oracle for a distribution $G''$ which is a smooth extension of
$G'.$ (See the discussion in Sections~4.1 and~4.2 of \cite{Jackson97}, in
particular Steps~16-18 of the {\tt HS} algorithm of Figure~4 and Steps~3 and~5
of the {\tt WDNF} algorithm of Figure~3.) By the definition of a smooth
extension, if $q$ is such that $G'[q]=0$ then $G''[q]$ also equals 0, and
consequently $g(q)=0$ as well.  Thus it is straightforward
to run the Harmonic Sieve using access to an approximate $G'$-restricted
evaluation oracle:  if $G'[q]$ returns $0$ then ``$0$'' is the correct value
of $g(q)$, and otherwise the oracle provides precisely the information that
would be available for the Sieve in Jackson's
original formulation.
\end{proof}

\subsection{$\AC^0$ queries using \cite{JacksonKS02}}
 \label{sec:AC0}

Fix $\U=\{0,1\}^d$ and $\Q=\{0,1\}^d.$  In this subsection we show that our reduction enables us to do
efficient private data release for quite a broad class of queries, namely
any query computed by a constant-depth circuit.

In more detail, let $P(q,u)\colon \{0,1\}^d \times \{0,1\}^d \to \{0,1\}$ be any
predicate that is computed by a circuit of depth~$\ell=O(1)$ and size $\poly(d).$
Our data release result for such queries is the following:

\begin{theorem}[Releasing $\AC^0$ queries]
\label{thm:releaseAC0}
Let $\GQ$ be the set containing the uniform distribution and let $\U,\Q,P$ be as
described above.
There is an $\eps$-differentially $(\U,\Q,\GQ,P)$ data release algorithm that is
    $(\alpha,\beta,\gamma)$-accurate and has runtime $\poly(n)$
for databases of size $n$, provided that
\[
n \geq d^{O\left(\log^\ell\left(\frac d{\alpha \gamma}\right)\right)} \cdot
\tilde{O}\left( {\frac{\log\left(\nfrac1\beta\right)^3}
{\epsilon\alpha^2 \gamma}}
\right)\mper
\]
\end{theorem}

See the introduction for a discussion of this result.
We observe that given any fixed $P$ as described above, for any given $u \in \U=\{0,1\}^d$ the function $p_u(q)$ is computed by a
circuit of depth~$\ell$ and size $\poly(d)$ over the input bits $q_1,\dots,q_d.$  Hence Theorem~\ref{thm:releaseAC0}
is an immediate consequence of Theorem~\ref{thm:main} and
the following learning result, which describes the performance guarantee of the quasipolynomial-time algorithm
of Jackson et al. \cite{JacksonKS02} for learning Majority-of-Parity
in our language:

\begin{theorem}[Theorem~9 of \cite{JacksonKS02}]
\label{thm:jks}
Let \begin{itemize}
\item $\U$ denote the data universe $\{0,1\}^d$;
\item $\Q$ denote the set of query descriptions $\{0,1\}^d$;
\item $P(q,u)$ be any fixed predicate computed by an AND/OR/NOT circuit of
depth $\ell=O(1)$ and size $\poly(d)$;
\item $\GQ$ contains only the uniform distribution over $\Q$; and
\item $\F$ be the set of all AND/OR/NOT circuits of depth~$\ell$ and size $\poly(d).$
\end{itemize}
Then there is an algorithm ${\cal L}$ that $(\gamma,\beta)$ learns $n$-thresholds over
$(\Q,\GQ',\F)$ where $\GQ'$ is the $(2/\gamma)$-smooth extension of $\GQ.$  Algorithm
 ${\cal L}$ uses
approximate distribution restricted oracle access  to the function, uses
 $b(n,\gamma,\beta)=d^{O(\log^\ell(nd/\gamma))} \cdot \log(1/\beta)$
 samples and calls to the evaluation oracle, and runs in time
 $t(n,\gamma,\beta) = d^{O(\log^\ell(nd/\gamma))} \cdot \log(1/\beta).$
\end{theorem}

We note that Theorem~9 of \cite{JacksonKS02}, as stated in that paper, only
deals with learning majority-of-$\AC^0$ circuits under the uniform
distribution: it says that an $n$-way Majority of depth-$\ell$, size-$\poly(d)$
circuits over $\{0,1\}^d$ can be learned to accuracy $\gamma$ and confidence
$\beta$ under the uniform distribution, using random examples only, in time
$d^{O(\log^\ell(nd/\gamma))}\cdot\log(1/\beta).$  However, the boosting-based
algorithm of \cite{JacksonKS02} is identical in its high-level structure to
Jackson's Harmonic Sieve; the only difference is that the \cite{JacksonKS02}
weak learner simply performs an exhaustive search over all low-weight parity
functions to find a weak hypothesis that has non-negligible correlation with
the target, whereas the Harmonic Sieve uses a more sophisticated
membership-query algorithm (that is an extension of the algorithm of
Kushilevitz and Mansour \cite{KushilevitzM93}).  Arguments identical to the
ones Jackson gives for the Harmonic Sieve (in Section~7.1 of \cite{Jackson97})
can be applied unchanged to the \cite{JacksonKS02} algorithm, to show that it
extends, just like the Harmonic Sieve, to learning under smooth distributions
if it is provided with an approximate evaluation oracle for the smooth
distribution.  In more detail, these arguments show that for any $C$-smooth
distribution $G'$, given sampling access to labeled examples by $(G',f)$
(where $f$ is the target $n$-way Majority of depth-$\ell$, size-$\poly(d)$
circuits) and approximate evaluation access to $G'$, the \cite{JacksonKS02}
algorithm learns $f$ to accuracy $\gamma$ and confidence $\beta$ under $G'$
in time $d^{O(\log^\ell(Cnd/\gamma))}\cdot\log(1/\beta)$  This is the result that
is restated in our data privacy language above (note that the smoothness
parameter there is $C=2/\gamma$).

\section{Conclusion and open problems}
\label{sec:openQ}

This work put forward a new reduction from privacy-preserving data analysis to
learning thresholds. Instantiating this reduction with various different
learning algorithms, we obtained new data release algorithms for a variety of
query classes. One notable improvement was for the database size (or error) in
distribution-free release of conjunctions and $k$-way conjunctions. Given
these new results, we see no known obstacles for even more dramatic
improvements on this central question. In particular, we conclude with the
following open question.

\begin{openQ}
Is there a differentially private distribution-free data release algorithm
(with constant error, e.g., $\alpha=1/100$) for conjunctions or $k$-way
conjunctions that works for databases of size $\poly(d)$ and runs in time
$\poly(n)$ (or $\poly(n,d^k)$ for the case of $k$-way conjunctions)?
\end{openQ}

Note that such an algorithm for $k$-way conjunctions would also imply, via boosting~\cite{DworkRV10}, that we can privately
release {\em all} $k$-way conjunctions in time $\poly(n,d^k)$, provided
that $|D|\ge \poly(d).$

\bibliographystyle{alpha}

\bibliography{moritz}

\appendix

\end{document}